\documentclass[11pt]{article}

\usepackage{mathpazo}
\usepackage{amsfonts}
\usepackage{amsmath}
\usepackage{graphicx}
\usepackage{latexsym}

\usepackage[margin=1.05in]{geometry}

\usepackage[T1]{fontenc}
\usepackage{times}
\usepackage{color,graphicx}
\usepackage{array}
\usepackage{enumerate}
\usepackage{amsmath}
\usepackage{amssymb}
\usepackage{amsthm}
\usepackage{pgfplots}
\usepackage{pgf}
\usepackage{tikz}
\usepackage{filecontents}
\usepgfplotslibrary{patchplots} 
\usetikzlibrary{pgfplots.patchplots} 
\pgfplotsset{width=9cm,compat=1.5.1}

\usepackage{amsthm}

\newtheorem{definition}{Definition}
\newtheorem{lemma}{Lemma}
\newtheorem{theorem}{Theorem}
\newtheorem{corollary}{Corollary}

\newtheorem{example}{Example}

\usepackage{xcolor}
\usepackage{makeidx}
\usepackage[colorlinks=true,linkcolor=blue,anchorcolor=blue,citecolor=red,urlcolor=magenta]{hyperref}

\newcommand{\tr}{\operatorname{Tr}}

\newcommand{\defeq}{\stackrel{\smash{\textnormal{\tiny def}}}{=}}

\begin{document}

\title{Pairwise Completely Positive Matrices and \\ Conjugate Local Diagonal Unitary Invariant Quantum States}

\author{
	Nathaniel Johnston\footnote{Department of Mathematics \& Computer Science, Mount Allison University, Sackville, NB, Canada E4L 1E4}\textsuperscript{$\ \ $}\footnote{Department of Mathematics \& Statistics, University of Guelph, Guelph, ON, Canada N1G 2W1}
	\quad and\quad
	Olivia MacLean\footnotemark[1]
}

\date{July 18, 2018}

\maketitle

\begin{abstract}
	We introduce a generalization of the set of completely positive matrices that we call ``pairwise completely positive'' (PCP) matrices. These are pairs of matrices that share a joint decomposition so that one of them is necessarily positive semidefinite while the other one is necessarily entrywise non-negative. We explore basic properties of these matrix pairs and develop several testable necessary and sufficient conditions that help determine whether or not a pair is PCP. We then establish a connection with quantum entanglement by showing that determining whether or not a pair of matrices is pairwise completely positive is equivalent to determining whether or not a certain type of quantum state, called a conjugate local diagonal unitary invariant state, is separable. Many of the most important quantum states in entanglement theory are of this type, including isotropic states, mixed Dicke states (up to partial transposition), and maximally correlated states. As a specific application of our results, we show that a wide family of states that have absolutely positive partial transpose are in fact separable.
\end{abstract}

\section{Introduction}

One of the most active areas of modern linear algebra research concerns \emph{completely positive (CP)} matrices \cite{BDS15}, which are matrices with a positive semidefinite decomposition consisting entirely of non-negative real numbers (or equivalently, matrices that are the Gram matrix of a set of vectors residing in the non-negative orthant of $\mathbb{R}^n$) \cite{Ber88,BS03}. Completely positive matrices, and their dual cone of \emph{copositive matrices}, have received so much attention lately primarily due to their role in mathematical optimization. In particular, many difficult optimization problems can be rephrased as linear programs with completely positive or copositive constraints, which shifts all of the difficulty of the optimization problem into the problem of understanding and characterizing these sets of matrices \cite{Bom12,Dur10}.

Very recently, it was shown that completely positive matrices also play an important role in quantum information theory, where one of the central research questions asks how to determine whether or not a quantum state is separable. It was shown in \cite{Yu16,TAQLS17} that there is a natural family of states called \emph{mixed Dicke states} whose separability is determined exactly by the complete positivity of a closely related matrix. Our contribution is to generalize the set of completely positive matrices to a set that we call \emph{pairwise completely positive} (PCP) matrices, establish the basic properties of these matrix pairs, show that these matrices are also connected to the separability problem in quantum information theory, and use that connection to make progress on a question about absolutely separable quantum states.

This paper is organized as follows. In Section~\ref{sec:math_prelims} we introduce the notation and mathematical preliminaries necessary to discuss these topics more formally. In Section~\ref{sec:ccp_matrices} we define pairwise completely positive matrices, which are the central subject of this paper, and prove some of their basic properties. Some of these basic properties of PCP matrices can be viewed as necessary conditions that matrix pairs must satisfy in order to be PCP, so in Section~\ref{sec:suff_conditions} we flip this around and introduce some sufficient conditions that \emph{guarantee} that a pair of matrices is PCP. In Section~\ref{sec:quantum_sep}, we show that determining whether or not pairs of matrices are PCP is equivalent to determining whether or not quantum states that are invariant under the action of local diagonal unitaries are separable. We present some examples to illustrate how our results apply to families of states like the isotropic states that are well-known. In Section~\ref{sec:abs_sep} we show how our results can be applied to the absolute separability problem, showing that a large family of states that are ``absolutely PPT'' are necessarily separable. Finally, we close in Section~\ref{sec:conclusions} by discussing some open problems are future directions for research coming out of this work.

\section{Mathematical Preliminaries}\label{sec:math_prelims}

We use bold letters like $\mathbf{v},\mathbf{w} \in \mathbb{C}^n$ to denote vectors, non-bold uppercase letters like $A,B \in M_n(\mathbb{C})$ to denote matrices, and non-bold lowercase letters like $c,d \in \mathbb{C}$ to denote scalars. Subscripts on non-bold letters indicate particular entries of a vector or matrix (e.g., $v_1,v_2,v_3$ denote the first $3$ coordinates of the vector $\mathbf{v}$), while subscripts of bold letters denote particular vectors (e.g., $\mathbf{v_1}, \mathbf{v_2}, \mathbf{v_3}$ are three vectors). Double subscripts are sometimes used to denote specific entries of vectors that are themselves denoted by subscripts (e.g., $v_{3,7}$ refers to the $7$-th entry of the vector $\mathbf{v_3}$). Sometimes it will be convenient to use square brackets to help us denote specific entries of a vector or matrix---for example, $[\mathbf{v}]_j := v_j$. We use $I_n$ to denote the $n \times n$ identity matrix and $\mathbf{1}$ to denote the all-ones vector: $\mathbf{1} := (1,1,\ldots,1)^T$.

The \emph{Hadamard product} of two vectors $\mathbf{v},\mathbf{w} \in \mathbb{C}^n$ or two matrices $A,B \in M_{m,n}(\mathbb{C})$ (denoted by $\mathbf{v}\odot\mathbf{w}$ or $A \odot B$) is simply their entrywise product: $[\mathbf{v} \odot \mathbf{w}]_j := v_jw_j$ for all $j$ and $[A \odot B]_{i,j} := a_{i,j}b_{i,j}$ for all $i,j$. We use standard inequality signs like ``$\geq$'' and ``$\leq$'' to denote entrywise inequalities between matrices, and we use ``$\succeq$'' and ``$\preceq$'' to denote inequalities in the Loewner partial order on matrices (e.g., $A \succeq B$ means that $A-B$ is positive semidefinite).

There are two matrix norms that will make frequent appearances throughout this work. The first of these is the \emph{entrywise $1$-norm} $\|\cdot\|_{1}$, defined simply via
\begin{align*}
	\|X\|_{1} \defeq \sum_{i,j=1}^n |x_{i,j}|,
\end{align*}
and the other one is the \emph{trace norm} $\|\cdot\|_{\textup{tr}}$, which is the sum of the singular values of the matrix:
\begin{align*}
	\|X\|_{\textup{tr}} \defeq \sum_{k=1}^n \sigma_k(A).
\end{align*}
However, there is another characterization of the trace norm that will also be useful for us, which is as the smallest possible rank-one decomposition of $X$:
\begin{align}\label{eq:tr_norm_as_inf}
	\|X\|_{\textup{tr}} = \inf\left\{ \sum_k \|\mathbf{x_k}\|\|\mathbf{y_k}\| : X = \sum_k \mathbf{x_k}\mathbf{y}_{\mathbf{k}}^*, \{\mathbf{x_k}\},\{\mathbf{y_k}\} \subseteq \mathbb{C}^n \right\},
\end{align}
where the infimum is taken over all rank-one sum decompositions of $X$. It is worth noting that the singular value decomposition of $X$ is in fact a rank-one sum decomposition that attains this infimum (i.e., choose $\{\mathbf{x_k}\}$ to be its left singular vectors and $\{\mathbf{y_k}\}$ to be its right singular vectors, suitably scaled by the associated singular values). Also, this characterization of the trace norm immediately gives the following simple relationship with the entrywise $1$-norm:
\begin{align}\label{eq:tr_norm_1_norm_ineq}
	\|X\|_{\textup{tr}} = \inf\left\{ \sum_k \|\mathbf{x_k}\|\|\mathbf{y_k}\| : X = \sum_k \mathbf{x_k}\mathbf{y_k}^* \right\} \leq \sum_{i,j=1}^n |x_{i,j}| = \|X\|_1,
\end{align}
where the inequality comes from choosing the na\"{i}ve rank-one sum decomposition of $X$ in terms of the standard basis vectors $\{\mathbf{e}_{\mathbf{i}}\}$:
\[
	X = \sum_{i,j=1}^n x_{i,j} \mathbf{e}_{\mathbf{i}}\mathbf{e}_{\mathbf{j}}^*.
\]

It follows that the quantity $\|X\|_1 - \|X\|_{\textup{tr}}$ is always non-negative. Furthermore, it can roughly be thought of as a measure of how far away from being diagonal $X$ is. For example, if $X$ is diagonal then $\|X\|_1 - \|X\|_{\textup{tr}} = 0$, whereas if $X$ is the all-ones matrix then $\|X\|_1 / \|X\|_{\textup{tr}} = n$, which is as large as possible. In fact, if $X$ is positive semidefinite then this quantity is simply the sum of the absolute values of the off-diagonal entries of $X$: $\|X\|_1 - \|X\|_{\textup{tr}} = \sum_{i \neq j}|x_{i,j}|$. This quantity is of interest in quantum information theory, and is called the $\ell_1$\emph{-norm of coherence} \cite{BCP14}.

\section{Definition and Basic Results}\label{sec:ccp_matrices}

A matrix $X \in M_n(\mathbb{R})$ is called \emph{completely positive} \cite{BS03} if there exists an entrywise non-negative matrix $V \in M_{n,m}(\mathbb{R})$ (with $m$ arbitrary) such that $X = VV^*$. Equivalently, there exist entrywise non-negative vectors $\{\mathbf{v_k}\} \subseteq \mathbb{R}^n$ (which are the columns of $V$) such that
\begin{align*}
	X = \sum_{k=1}^m \mathbf{v_k}\mathbf{v}_{\mathbf{k}}^*.
\end{align*}
In particular, completely positive matrices are (by construction) both positive semidefinite and entrywise non-negative. The main purpose of this paper is to introduce and explore the following generalization of this concept:

\begin{definition}\label{lem:ccp_matrices}
	Suppose $X,Y \in M_n(\mathbb{C})$. We say that the pair $(X,Y)$ is \emph{pairwise completely positive} (or \emph{PCP} for short) if there exist matrices $V,W \in M_{n,m}(\mathbb{C})$ (with $m$ arbitrary) such that
    \[
	   X = (V \odot W)(V \odot W)^* \qquad \text{and} \qquad Y = (V \odot \overline{V})(W \odot \overline{W})^*.
	\]
	Equivalently, $(X,Y)$ is PCP if there exist families of vectors $\{\mathbf{v_k}\},\{\mathbf{w_k}\} \subseteq \mathbb{C}^n$ (which are the columns of $V$ and $W$, respectively) such that
	\begin{align}\label{eq:ccp_defn_sum}
		X = \sum_{k=1}^m \mathbf{(v_k\odot w_k)}\mathbf{(v_k\odot w_k)}^* \qquad \text{and} \qquad Y = \sum_{k=1}^m \mathbf{(v_k\odot \overline{v_k})}\mathbf{(w_k\odot\overline{w_k})}^*.
	\end{align}
\end{definition}

In a sense, each of the matrices $X$ and $Y$ in the above definition captures a ``piece'' of the complete positivity property. While $X$ is (by construction) positive semidefinite, it need not be entrywise non-negative, and similarly $Y$ is clearly entrywise non-negative, but is not necessarily positive semidefinite (or even symmetric). There are also a few other easily-checkable properties that $X$ and $Y$ must satisfy in order to have a chance at being pairwise completely positive, some of which are not obvious. We summarize these properties in the following theorem.

\begin{theorem}\label{thm:ccp_properties}
	Suppose $X,Y \in M_n(\mathbb{C})$ are such that $(X,Y)$ is pairwise completely positive. Then each of the following conditions hold:
	\begin{enumerate}
		\item[a)] $X$ is (Hermitian) positive semidefinite.
		
		\item[b)] $Y$ is real and entrywise non-negative.
		
		\item[c)] $X$ and $Y$ have the same diagonal entries: $x_{i,i} = y_{i,i}$ for all $1 \leq i \leq n$.
		
		\item[d)] $Y$ is ``almost'' entrywise larger than $X$: $|x_{i,j}|^2 \leq y_{i,j}y_{j,i}$ for all $1 \leq i,j \leq n$.
		
		\item[e)] $X$ is ``more diagonal'' than $Y$: $\|X\|_1 - \|X\|_{\textup{tr}} \leq \|Y\|_1 - \|Y\|_{\textup{tr}}$.
	\end{enumerate}
\end{theorem}

Before proving the above result, it is perhaps worth noting that if $Y$ is symmetric then part~(d) of this theorem says exactly that $Y$ is entrywise larger than $X$: $y_{i,j} \geq |x_{i,j}|$ for all $1 \leq i,j \leq n$. Also, Inequality~\eqref{eq:tr_norm_1_norm_ineq} tells us that the quantities $\|X\|_1 - \|X\|_{\textup{tr}}$ and $\|Y\|_1 - \|Y\|_{\textup{tr}}$ in part~(e) of the theorem are both non-negative. Also, since $X$ is necessarily positive semidefinite, we have $\|X\|_{\textup{tr}} = \tr(X)$ (but no similar simplification can be made for $\|Y\|_{\textup{tr}}$ in general).

\begin{proof}
	We already noted properties~(a) and~(b), and they follow straightforwardly from the definition of PCP matrices.
	
	For property~(c), note that if $(X,Y)$ is PCP then we can compute the diagonal entries of $X$ and $Y$ in terms of the entries of the vectors $\{\mathbf{v_k}\}$ and $\{\mathbf{w_k}\}$:
	\begin{align*}
		x_{i,i} & = \sum_k \big[ \mathbf{(v_k\odot w_k)}\mathbf{(v_k\odot w_k)}^* \big]_{i,i} = \sum_k |v_{k,i}|^2|w_{k,i}|^2 \quad \text{and} \\
		y_{i,i} & = \sum_k \big[ \mathbf{(v_k\odot \overline{v_k})}\mathbf{(w_k\odot\overline{w_k})}^* \big]_{i,i} = \sum_k |v_{k,i}|^2|w_{k,i}|^2,
	\end{align*}
	which are equal to each other.
	
	For property~(d), we again directly compute the relevant entries of $X$ and $Y$ in terms of the entries of $\{\mathbf{v_k}\}$ and $\{\mathbf{w_k}\}$:
	\begin{align*}
		x_{i,j} & = \sum_k \big[ \mathbf{(v_k\odot w_k)}\mathbf{(v_k\odot w_k)}^* \big]_{i,j} = \sum_k v_{k,i}w_{k,i}\overline{v_{k,j}w_{k,j}} \quad \text{and} \\
		y_{i,j} & = \sum_k \big[ \mathbf{(v_k\odot \overline{v_k})}\mathbf{(w_k\odot\overline{w_k})}^* \big]_{i,j} = \sum_k |v_{k,i}|^2|w_{k,j}|^2.
	\end{align*}
	Then the fact that $|x_{i,j}|^2 \leq y_{i,j}y_{j,i}$ is then simply a result of applying the Cauchy--Schwarz inequality to the vectors
	\begin{align*}
		(v_{1,i}\overline{w_{1,j}}, v_{2,i}\overline{w_{2,j}}, v_{3,i}\overline{w_{3,j}}, \ldots) \quad \text{and} \quad (w_{1,i}\overline{v_{1,j}}, w_{2,i}\overline{v_{2,j}}, w_{3,i}\overline{v_{3,j}}, \ldots).
	\end{align*}
	
	For property~(e), we use Equation~\eqref{eq:tr_norm_as_inf} to bound $\|Y\|_{\textup{tr}}$ as follows:
	\begin{align*}
		\|Y\|_{\textup{tr}} & \leq \sum_k \|\mathbf{v_k\odot \overline{v_k}}\|\|\mathbf{w_k\odot \overline{w_k}}\|.
	\end{align*}
	Thus
	\begin{align*}
	    \|Y\|_1 - \|Y\|_{\textup{tr}} & \geq \left\|\sum_k (\mathbf{v_k\odot \overline{v_k}})(\mathbf{w_k\odot \overline{w_k}})^*\right\|_1 - \sum_k \|\mathbf{v_k\odot \overline{v_k}}\|\|\mathbf{w_k\odot \overline{w_k}}\| \\
	    & = \sum_k \|\mathbf{v_k\odot \overline{v_k}}\|_1\|\mathbf{w_k\odot \overline{w_k}}\|_1 - \sum_k \|\mathbf{v_k\odot \overline{v_k}}\|\|\mathbf{w_k\odot \overline{w_k}}\|,
	\end{align*}
	where the final equality follows from that fact that each term in the sum is entrywise non-negative.
	
	On the other hand,
	\begin{align*}
		\|X\|_1 - \|X\|_{\textup{tr}} & = \left\|\sum_{k} \mathbf{(v_k\odot w_k)}\mathbf{(v_k\odot w_k)}^*\right\|_1 - \sum_{k} \|\mathbf{(v_k\odot w_k)}\|^2 \\
		& \leq \sum_{k} \|\mathbf{(v_k\odot w_k)}\|_1^2 - \sum_{k} \|\mathbf{(v_k\odot w_k)}\|^2,
	\end{align*}
	where the final inequality is just the triangle inequality.
	
	By comparing these expressions for $\|Y\|_1 - \|Y\|_{\textup{tr}}$ and $\|X\|_1 - \|X\|_{\textup{tr}}$, we see that it suffices to prove that
	\begin{align*}
	    \|\mathbf{(v\odot w)}\|_1^2 - \|\mathbf{(v\odot w)}\|^2 \leq \|\mathbf{v\odot \overline{v}}\|_1\|\mathbf{w\odot \overline{w}}\|_1 - \|\mathbf{v\odot \overline{v}}\|\|\mathbf{w\odot \overline{w}}\|
	\end{align*}
	for all $\mathbf{v},\mathbf{w} \in \mathbb{C}^n$. By using the easily-verified facts that $\|\mathbf{(v\odot w)}\|^2 = \langle \mathbf{v\odot \overline{v}}, \mathbf{w\odot \overline{w}} \rangle$, $\|\mathbf{v\odot \overline{v}}\|_1 = \|\mathbf{v}\|^2$, $\|\mathbf{w\odot \overline{w}}\|_1 = \|\mathbf{w}\|^2$, $\|\mathbf{(v\odot w)}\|_1 \geq |\langle \mathbf{v}, \mathbf{w} \rangle|$, and rearranging slightly, we see that it suffices to show that
	\begin{align}\label{eq:stronger_cs_ineq}
	    \|\mathbf{v\odot \overline{v}}\|\|\mathbf{w\odot \overline{w}}\| - \langle \mathbf{v\odot \overline{v}}, \mathbf{w\odot \overline{w}} \rangle \leq \|\mathbf{v}\|^2\|\mathbf{w}\|^2 - |\langle \mathbf{v}, \mathbf{w} \rangle|^2
	\end{align}
	for all $\mathbf{v},\mathbf{w} \in \mathbb{C}^n$. The proof of this inequality is somewhat involved, so we leave it to the appendix, as Lemma~\ref{lem:cs_stronger}.
\end{proof}

It is worth noting that the Inequality~\eqref{eq:stronger_cs_ineq} (i.e., Lemma~\ref{lem:cs_stronger} in the appendix) is perhaps of independent interest, as it is a non-trivial strengthening of the Cauchy--Schwarz inequality. Indeed, the Cauchy--Schwarz inequality itself tells us that the left-hand-side of~\eqref{eq:stronger_cs_ineq} is non-negative, so~\eqref{eq:stronger_cs_ineq} gives a non-negative lower bound on $\|\mathbf{v}\|^2\|\mathbf{w}\|^2 - |\langle \mathbf{v}, \mathbf{w} \rangle|^2$ (whereas the Cauchy--Schwarz inequality only gives a lower bound of $0$ on it).

\begin{example}\label{exam:3dim_eg}
	Let $a > 0$ be a real number and consider the pair of matrices
	\begin{align*}
	X = \begin{bmatrix}
	1 & 1 & 1 \\
	1 & 1 & 1 \\
	1 & 1 & 1
	\end{bmatrix} \quad \text{and} \quad Y = \begin{bmatrix}
		1 & a & 1/a \\
		1/a & 1 & a \\
		a & 1/a & 1
	\end{bmatrix}.
	\end{align*}
	We now determine which values of $a$ result in $(X,Y)$ being PCP.
	
	It is straightforward to show that $(X,Y)$ satisfies properties (a), (b), (c), and (d) of Theorem~\ref{thm:ccp_properties}, regardless of the value of $a > 0$. However, direct calculation shows that
	\begin{align*}
		\|X\|_1 = 9, \ \ \|X\|_{\textup{tr}} = 3, \ \ \|Y\|_1 = \frac{3(1 + a + a^2)}{a}, \quad \text{and} \quad \|Y\|_{\textup{tr}} = \frac{1 + a + a^2 + 2|a-1|\sqrt{1 + a + a^2}}{a}.
	\end{align*}
	Thus $\|X\|_1 - \|X\|_{\textup{tr}} = 9 -3 = 6$, whereas
	\begin{align*}
		\|Y\|_1 - \|Y\|_{\textup{tr}} = \frac{2(1 + a + a^2 - |a-1|\sqrt{1 + a + a^2})}{a},
	\end{align*}
	which is strictly less than $6$ whenever $a \neq 1$. Thus Theorem~\ref{thm:ccp_properties}(e) tells us that $(X,Y)$ is not PCP when $a \neq 1$. On the other hand, $(X,Y)$ \emph{is} PCP when $a = 1$, since in this case there is a trivial PCP decomposition of $(X,Y)$: $\mathbf{v_1} = \mathbf{w_1} = (1,1,1)^T$.
\end{example}

In the previous example, we saw that $(X,Y)$ was pairwise completely positive when $X = Y$ was completely positive. The following theorem shows that this fact holds in general, and thus pairwise completely positive matrices are indeed a generalization of completely positive matrices.

\begin{theorem}\label{thm:cp_is_ccp}
	Suppose $X \in M_n(\mathbb{C})$. Then $X$ is completely positive if and only if $(X,X)$ is pairwise completely positive.
\end{theorem}

\begin{proof}
	For the ``only if'' direction, suppose $X$ is completely positive. Then there exist entrywise non-negative vectors $\{\mathbf{x_k}\} \subseteq \mathbb{C}^n$ such that
	\begin{align*}
		X = \sum_k \mathbf{x_k}\mathbf{x}_{\mathbf{k}}^*.
	\end{align*}
	Well, if we define $\mathbf{v_k} = \mathbf{w_k} = \sqrt{\mathbf{x_k}}$ (where we understand that this square root is taken entrywise), then it is straightforward to verify that
	\begin{align*}
		X = \sum_k \mathbf{(v_k\odot w_k)}\mathbf{(v_k\odot w_k)}^* \qquad \text{and} \qquad X = \sum_k \mathbf{(v_k\odot \overline{v_k})}\mathbf{(w_k\odot\overline{w_k})}^*,
	\end{align*}
	so $(X,X)$ is pairwise completely positive.
	
	For the ``if'' direction, suppose $(X,X)$ is PCP, so
	\[
	    X = \sum_k \mathbf{(v_k\odot w_k)}\mathbf{(v_k\odot w_k)}^* = \sum_k \mathbf{(v_k\odot \overline{v_k})}\mathbf{(w_k\odot\overline{w_k})}^*.
	\]
	Then for each $i \neq j$, the $(i,j)$-entry of $X$ is
	\begin{align*}
	    x_{i,j} = \sum_k v_{i,k}w_{i,k}\overline{v_{j,k}w_{j,k}} = \sum_k |v_{i,k}|^2|w_{j,k}|^2, \quad \text{so} \quad \sum_k v_{i,k}\overline{w_{j,k}}\big(\overline{v_{i,k}}w_{j,k} - \overline{v_{j,k}}w_{i,k}\big) = 0.
	\end{align*}
	If we define vectors $\mathbf{y}_{\mathbf{i,j}} = (v_{i,1}\overline{w_{j,1}},v_{i,2}\overline{w_{j,2}},v_{i,3}\overline{w_{j,3}},\ldots)$, then the above equality simply says that
	\begin{align*}
	    \big\langle \mathbf{y}_{\mathbf{i,j}}, \mathbf{y}_{\mathbf{i,j}} - \mathbf{y}_{\mathbf{j,i}} \big\rangle = 0 \quad \text{for all} \quad i \neq j.
	\end{align*}
	We thus conclude that $\langle \mathbf{y}_{\mathbf{i,j}}, \mathbf{y}_{\mathbf{j,i}} \rangle = \|\mathbf{y}_{\mathbf{i,j}}\|^2$, and a similar argument shows that $\langle \mathbf{y}_{\mathbf{i,j}}, \mathbf{y}_{\mathbf{j,i}} \rangle = \|\mathbf{y}_{\mathbf{j,i}}\|^2$ as well. The equality condition of the Cauchy--Schwarz inequality then tells us that $\mathbf{y}_{\mathbf{i,j}} = \mathbf{y}_{\mathbf{j,i}}$ for all $i \neq j$. In other words, $v_{i,k}\overline{w_{j,k}} = v_{j,k}\overline{w_{i,k}}$ for all $i,j,k$. Thus $\mathbf{v}_{\mathbf{k}}$ is proportional to $\overline{\mathbf{w}_{\mathbf{k}}}$ for all $k$; there exist scalars $\{c_k\} \subset \mathbb{C}$ such that $\mathbf{v}_{\mathbf{k}} = c_k\overline{\mathbf{w}_{\mathbf{k}}}$ for all $k$. Thus
	\[
	    X = \sum_k \mathbf{(v_k\odot w_k)}\mathbf{(v_k\odot w_k)}^* = \sum_k (c_k\overline{\mathbf{w}_{\mathbf{k}}}\odot \mathbf{w}_{\mathbf{k}})(c_k\overline{\mathbf{w}_{\mathbf{k}}}\odot \mathbf{w}_{\mathbf{k}})^* = \sum_k |c_k|^2(\overline{\mathbf{w}_{\mathbf{k}}}\odot \mathbf{w}_{\mathbf{k}})(\overline{\mathbf{w}_{\mathbf{k}}}\odot \mathbf{w}_{\mathbf{k}})^*,
	\]
	so $X$ is completely positive, since each vector $\overline{\mathbf{w}_{\mathbf{k}}}\odot \mathbf{w}_{\mathbf{k}}$ is entrywise non-negative.
\end{proof}

Since it is NP-hard to check whether or not a matrix $X$ is completely positive \cite{DG14}, the above result immediately implies that it is also NP-hard to check pairwise complete positivity of a pair $(X,Y)$. The following example illustrates some pairwise completely positive matrices for which $X \neq Y$.

\begin{example}\label{exam:werner_ccp}
	Let $a,b \in \mathbb{R}$ and let $J_n \in M_n(\mathbb{C})$ be the matrix with all of its entries equal to $1$. Consider the pair of matrices
	\begin{align*}
	    X = aI_n + bJ_n \quad \text{and} \quad Y = bI_n + aJ_n.
	\end{align*}
	In order for the pair $(X,Y)$ to satisfy property~(a) of Theorem~\ref{thm:ccp_properties} we need $a \geq 0$ and $b \geq -a/n$, and for it to satisfy property~(d) we need $b \leq a$.
	
	On the other hand, we now show that $(X,Y)$ \emph{is} PCP when these inequalities are satisfied (i.e., when $a \geq 0$ and $-a/n \leq b \leq a$). To this end, assume without loss of generality that $a = n$ (otherwise we can just rescale $X$ and $Y$ so that this is the case). We now construct explicit PCP decompositions of $(X,Y)$ in the $b = -1$ and $b = n$ cases, and pairwise complete positivity on the entire interval $-1 \leq b \leq n$ then follows from convexity.
	
	For the $b = -1$ case, we first define the quantities
	\begin{align*}
	    c_{\pm} = \sqrt{\frac{1}{2}\Big(n^2 - n + 2 \pm \sqrt{n^4 - 2n^3 + n^2 + 4n} \Big)}.
	\end{align*}
	It is straightforward to check that these quantities are real for all $n \geq 1$ (i.e., the quantities under square roots are non-negative) and that $c_{+}c_{-} = n-1$ and $c_{+}^2 + c_{-}^2 = n^2 - n + 2$. If we define the vectors
	\begin{align*}
	    \mathbf{v_k} & = \frac{1}{\sqrt{n}}\big((c_{+}-1)\mathbf{e}_k + \mathbf{1}\big) \quad \text{and} \quad \mathbf{w_k} = (c_{-}+1)\mathbf{e}_k - \mathbf{1} \quad \text{for all} \quad 1 \leq k \leq n
	\end{align*}
	then it is straightforward to check that $X$ and $Y$ have the CCP decomposition
	\[
	    X = \sum_{k=1}^n\mathbf{(v_k\odot w_k)}\mathbf{(v_k\odot w_k)}^* \qquad \text{and} \qquad Y = \sum_{k=1}^n \mathbf{(v_k\odot \overline{v_k})}\mathbf{(w_k\odot\overline{w_k})}^*.
	\]
	
	In the $b = n$ case, we note that $X = Y$ so by Theorem~\ref{thm:cp_is_ccp} it suffices to show that $X$ is completely positive. Indeed, it has CP decomposition
    \[
        X = \sum_{k=1}^n \mathbf{v}_{\mathbf{k}}\mathbf{v}_{\mathbf{k}}^* \quad \text{with} \quad \mathbf{v}_{\mathbf{k}} = \sqrt{\frac{n}{n+2+2\sqrt{n+1}}} \big(\mathbf{1} + (1+\sqrt{n+1})\mathbf{e}_{\mathbf{k}}\big) \quad \text{for all} \quad 1 \leq k \leq n.
    \]
\end{example}

It is also worth noting a few ways in which we can modify or combine PCP matrix pairs without breaking the PCP property. In particular, if $D$ is an entrywise non-negative diagonal matrix and $P$ is a permutation matrix then:\smallskip
\begin{itemize}
    \item If $(X,Y)$ is PCP then so are $(DXD^*,DYD^*)$ and $(PXP^*,PYP^*)$.
    
    \item If $(X_1,Y_1)$ and $(X_2,Y_2)$ are PCP then so is $(X_1+X_2,Y_1+Y_2)$.\smallskip
\end{itemize}
The above properties all follow immediately from looking at how the given transformation affects the PCP decomposition of the given pair. There are also a few special cases where it is straightforward to check that a pair $(X,Y)$ is pairwise completely positive, and we catalog some of these simple cases here for easy reference.

\begin{lemma}\label{lem:ccp_diagonal}
	Suppose $X,Y \in M_n(\mathbb{C})$ satisfy conditions~(a)--(c) of Theorem~\ref{thm:ccp_properties}. If $X$ is diagonal then $(X,Y)$ is pairwise completely positive.
\end{lemma}

\begin{proof}
    Define families of vectors $\{\mathbf{v_{i,j}}\}$ and $\{\mathbf{w_{i,j}}\}$ by $\mathbf{v_{i,j}} = \mathbf{e_i}$ and $\mathbf{w_{i,j}} = \sqrt{y_{i,j}}\mathbf{e_j}$ for all $1 \leq i,j, \leq n$. It is then straightforward to verify that
    \begin{align*}
        \sum_{i,j=1}^n \mathbf{(v_{i,j}\odot w_{i,j})}\mathbf{(v_{i,j}\odot w_{i,j})}^* & = \sum_{j=1}^n (\sqrt{y_{j,j}}\mathbf{e_j})(\sqrt{y_{j,j}}\mathbf{e_j})^* = X \quad \text{and} \\
        \sum_{i,j=1}^n (\mathbf{v_{i,j}}\odot \overline{\mathbf{v_{i,j}}})(\mathbf{w_{i,j}}\odot \overline{\mathbf{w_{i,j}}})^* & = \sum_{i,j=1}^n \mathbf{e_i}(y_{i,j}\mathbf{e_j})^* = Y,
    \end{align*}
    which is a PCP decomposition of $(X,Y)$.
\end{proof}

The above lemma, as well as properties~(d) and~(e) of Theorem~\ref{thm:ccp_properties}, all roughly say that if $(X,Y)$ is PCP then the off-diagonal portion of $Y$ should be larger than that of $X$. The following lemma provides yet another way of making this idea precise, and acts as a sort of converse to Lemma~\ref{lem:ccp_diagonal}.

\begin{lemma}\label{lem:ccp_diagonal_b}
	Suppose $X,Y \in M_n(\mathbb{C})$ satisfy conditions~(a)--(c) of Theorem~\ref{thm:ccp_properties} and furthermore that $Y$ is diagonal. Then $(X,Y)$ is pairwise completely positive if and only if $X$ is diagonal.
\end{lemma}

\begin{proof}
    The ``if'' direction of this result follows immediately from Lemma~\ref{lem:ccp_diagonal}. For the ``only if'' direction, suppose $(X,Y)$ is PCP. Then Theorem~\ref{thm:ccp_properties}(d) tells us that $|x_{i,j}|^2 \leq y_{i,j}y_{j,i}$ for all $1 \leq i,j \leq n$. Since $Y$ is diagonal, we know that $y_{i,j} = 0$ whenever $i \neq j$, so $x_{i,j} = 0$ when $i \neq j$ as well (i.e., $X$ is also diagonal).
\end{proof}

As yet another way of making this idea precise that the off-diagonal entries of $Y$ in any PCP pair should be ``large'' compared to the off-diagonal entries of $X$, we note that increasing the off-diagonal entries of $Y$ can never break the PCP property, nor can increasing the diagonal entries of $X$ and $Y$ by the same amount:

\begin{lemma}\label{lem:ccp_increate_y_offdiag}
	Suppose $X,Y,P \in M_n(\mathbb{C})$ are such that $(X,Y)$ is pairwise completely positive and $P$ is entrywise non-negative. Then $(X+\operatorname{diag}(P),Y+P)$ is also pairwise completely positive.
\end{lemma}

\begin{proof}
    We just note that $(\operatorname{diag}(P),P)$ is PCP by Lemma~\ref{lem:ccp_diagonal}, and the sum of two PCP pairs is again PCP, so $(X,Y) + (\operatorname{diag}(P),P) = (X+\operatorname{diag}(P),Y+P)$ is PCP as well.
\end{proof}

As an immediately corollary, we get the following slight strengthening of Theorem~\ref{thm:cp_is_ccp}:

\begin{corollary}\label{cor:Y_bigger_than_X_CCP}
    Suppose $X,Y \in M_n(\mathbb{C})$ have the same diagonal entries and are such that $X$ is completely positive and $Y \geq X$ (where this inequality is meant entrywise). Then $(X,Y)$ is pairwise completely positive.
\end{corollary}

\begin{proof}
    Follows immediately from combining Lemma~\ref{lem:ccp_increate_y_offdiag} with Theorem~\ref{thm:cp_is_ccp}.
\end{proof}

\section{Sufficient Conditions}\label{sec:suff_conditions}

Theorem~\ref{thm:ccp_properties} provides several necessary conditions that must be satisfied in order for a pair $(X,Y)$ to have a chance of being PCP. In this section, we instead flip this around and present some \emph{sufficient} conditions that can be used to show that a pair must be PCP. Keep in mind that the problem of determining whether or not a matrix pair is PCP is NP-hard, so we do not expect to find any computationally simple conditions that are both necessary \emph{and} sufficient.

For our first nontrivial sufficient condition, recall that there is a simple characterization of completely positive matrices in small dimensions: when $n \leq 4$ then $X$ is CP if and only if it is doubly non-negative (see \cite{BS03,MM62}, for example). An analogous result holds for PCP matrices if $n = 2$.

\begin{theorem}\label{thm:ccp_n2}
	Suppose $X,Y \in M_2(\mathbb{C})$. Then $(X,Y)$ is pairwise completely positive if and only if conditions (a)--(d) of Theorem~\ref{thm:ccp_properties} hold.
\end{theorem}

\begin{proof}
    Theorem~\ref{thm:ccp_properties} already establishes the ``only if'' direction, so we only need to prove the ``if'' direction. With this in mind, suppose properties (a)--(d) of Theorem~\ref{thm:ccp_properties} hold and write
    \[
        X= \begin{bmatrix}
            x_{1,1} & \overline{x_{2,1}} \\
            x_{2,1} & x_{2,2}
        \end{bmatrix}, \quad Y = \begin{bmatrix}
            x_{1,1} & y_{1,2} \\
            y_{2,1} & x_{2,2}
        \end{bmatrix}.
    \]
    If $x_{1,1} = 0$ then $x_{1,2} = x_{2,1} = 0$ by positive semidefiniteness of $X$, so $X$ is diagonal and thus Lemma~\ref{lem:ccp_diagonal} shows that $(X,Y)$ is PCP. Similarly, if $y_{1,2} = 0$ then property~(d) of Theorem~\ref{thm:ccp_properties} tells us that $x_{1,2} = x_{2,1} = 0$, so again $X$ is diagonal and $(X,Y)$ is PCP. We thus assume from now on that $x_{1,1} \neq 0$ and $y_{1,2} \neq 0$, and our goal is to show that there exist families of vectors $\{ \mathbf{v_j}\},\{\mathbf{w_j}\} \subseteq \mathbb{C}^{2} $ such that
    \begin{align*}
        X = \sum_j(\mathbf{v_{j}\odot w_{j}})(\mathbf{v_{j}\odot w_{j}})^{*} \quad \text{and} \quad Y = \sum_j(\mathbf{v_{j}\odot \overline{v_{j}}})(\mathbf{w_{j}\odot \overline{w_{j}}})^{*}.
    \end{align*}
    To this end, define vectors $\mathbf{v_1}, \mathbf{v_2},\mathbf{w_1},\mathbf{w_2}$ in terms of the entries of $X$ and $Y$ as follows:
    \begin{align*}
        \mathbf{v_1} = \begin{bmatrix} 1 \\ \frac{x_{2,1}}{\sqrt{x_{1,1}y_{1,2}}}\end{bmatrix}, \quad \mathbf{v_2} = \begin{bmatrix} 0 \\ 1 \end{bmatrix}, \quad \mathbf{w_1} = \begin{bmatrix} \sqrt{x_{1,1}} \\ \sqrt{y_{1,2}}\end{bmatrix}, \quad \text{and} \quad \mathbf{w_2} = \begin{bmatrix}\sqrt{y_{2,1}-\frac{|x_{2,1}|^2}{y_{1,2}}} \\ \sqrt{x_{2,2}-\frac{|x_{2,1}|^2}{x_{1,1}}}\\
        \end{bmatrix},
    \end{align*}
    where we note that positive semidefiniteness of $X$ ensures that $x_{2,2}-\frac{|x_{2,1}|^2}{x_{1,1}} \geq 0$, and property~(d) of Theorem~\ref{thm:ccp_properties} ensures that $y_{2,1}-\frac{|x_{2,1}|^2}{y_{1,2}} \geq 0$.
    
    Then direct computation shows that
    \begin{align*}
        \sum_{j=1}^2(\mathbf{v_{j}\odot w_{j}})(\mathbf{v_{j}\odot w_{j}})^{*} & = \begin{bmatrix}
            x_{1,1} & \overline{x_{2,1}} \\
            x_{2,1} & |x_{2,1}|^2\over{x_{1,1}}
        \end{bmatrix} 
       + \begin{bmatrix}
            0 & 0 \\
            0 & x_{2,2}-\frac{|x_{2,1}|^2}{x_{1,1}}
        \end{bmatrix}=X, \quad \text{and} \\
     \noindent   
     \sum_{j=1}^2(\mathbf{v_{j}\odot \overline{v_{j}}})(\mathbf{w_{j}\odot \overline{w_{j}}})^{*} & =   
       \begin{bmatrix}
            x_{1,1} & y_{1,2}\\
            {|x_{2,1}|^2}\over{y_{1,2}} & |x_{2,1}|^2\over{x_{1,1}}
        \end{bmatrix} 
       + \begin{bmatrix}
            0 & 0 \\
            y_{2,1}-\frac{|x_{2,1}|^2}{y_{1,2}} & x_{2,2}-\frac{|x_{2,1}|^2}{x_{1,1}}
        \end{bmatrix}=Y,
\end{align*}
\noindent
as desired.
\end{proof}

On the other hand, Example~\ref{exam:3dim_eg} shows that the above characterization of PCP matrices does not hold for matrices $X,Y \in M_3(\mathbb{C})$. We are not able to get a complete necessary and sufficient condition in dimension $3$ or larger, but we can prove the following sufficient condition that reduces to that of the above theorem when $n = 2$.

\begin{theorem}\label{thm:ccp_sufficient_large}
	Suppose $X,Y \in M_n(\mathbb{C})$ satisfy properties (a)--(c) of Theorem~\ref{thm:ccp_properties}. For each $k = 1, 2, \ldots, n$ define the vectors $\mathbf{v_k}, \mathbf{w_k} \in \mathbb{C}^n$ recursively via the following formulas for their entries:
    \begin{align*}
	    v_{k,j} & = \begin{cases}
	        0 & \text{if} \ \ 1 \leq j < k \\
	        \displaystyle\frac{x_{j,k} - \sum_{i=1}^{k-1}v_{i,j}w_{i,j}\overline{v_{i,k}w_{i,k}}}{\sqrt{y_{k,j}- \sum_{i=1}^{k-1}|v_{i,k}|^2|w_{i,j}|^2 }} & \text{if} \ \ k \leq j \leq n
	    \end{cases} \qquad \text{and} \qquad w_{k,j} = \frac{\sqrt{y_{k,j} - \sum_{i=1}^{k-1}|v_{i,k}|^2|w_{i,j}|^2}}{v_{k,k}}.
	\end{align*}
	As long as all of the quantities under square roots above are non-negative and quantities in denominators are non-zero, the pair $(X,Y)$ is pairwise completely positive via the decomposition
	\[
	    X = \sum_{k=1}^n \mathbf{(v_k\odot w_k)}\mathbf{(v_k\odot w_k)}^* \qquad \text{and} \qquad Y = \sum_{k=1}^n \mathbf{(v_k\odot \overline{v_k})}\mathbf{(w_k\odot\overline{w_k})}^*.
	\]
\end{theorem}

\begin{proof}
    We just need to show that that claimed decompositions do indeed give us $X$ and $Y$, which just boils down to some somewhat tedious computation. For simplicity of notation, we start by defining the quantities
    \begin{align*}
        c_{j,k} = \sum_{i=1}^{k-1}v_{i,j}w_{i,j}\overline{v_{i,k}w_{i,k}} \quad \text{and} \quad d_{k,j} = \sum_{i=1}^{k-1}|v_{i,k}|^2|w_{i,j}|^2,
    \end{align*}
    so what we can express $x_{k,j}$ and $y_{k,j}$ more simply as
    \begin{align}\label{eq:vk_wk_compact}
         v_{k,j} & = \begin{cases}
    	        0 & \text{if} \ \ 1 \leq j < k \\
    	        \displaystyle\frac{x_{j,k} - c_{j,k}}{\sqrt{y_{k,j} - d_{k,j}}} & \text{if} \ \ k \leq j \leq n \qquad \text{and}
    	   \end{cases} \qquad w_{k,j} = \frac{\sqrt{y_{k,j} - d_{k,j}}}{v_{k,k}}.
    \end{align}
	
    Let's start by computing the $(j,k)$-entry of the claimed decomposition of $Y$:
    \begin{align*}
        \left[\sum_{i=1}^n \mathbf{(v_i\odot \overline{v_i})}\mathbf{(w_i\odot\overline{w_i})}^*\right]_{j,k} & = \sum_{i=1}^n |v_{i,j}|^2|w_{i,k}|^2 = \sum_{i=1}^j |v_{i,j}|^2|w_{i,k}|^2 = \left(\sum_{i=1}^{j-1} |v_{i,j}|^2|w_{i,k}|^2\right) + |v_{j,j}|^2|w_{j,k}|^2 \\
        & = d_{j,k} + |v_{j,j}|^2|w_{j,k}|^2 = d_{j,k} + |v_{j,j}|^2\left(\frac{y_{j,k} - d_{j,k}}{|v_{j,j}|^2}\right) = y_{j,k},
    \end{align*}
    where the second equality follows from the fact $v_{i,j} = 0$ whenever $i > j$ and the second-to-last equality follows from plugging in the formula~\eqref{eq:vk_wk_compact} for $w_{j,k}$.
    
    Thus this is indeed a valid PCP decomposition of $Y$. To see that it similarly gives the correct $X$, we similarly compute the $(j,k)$-entry of the claimed decomposition of $X$:
    \begin{align}\begin{split}\label{eq:start_of_X_suff_decomp}
        \left[\sum_{i=1}^n \mathbf{(v_i\odot w_i)}\mathbf{(v_i\odot w_i)}^*\right]_{j,k} & = \sum_{i=1}^n v_{i,j}w_{i,j}\overline{v_{i,k}w_{i,k}} = \sum_{i=1}^{k} v_{i,j}w_{i,j}\overline{v_{i,k}w_{i,k}} \\
        & = \left(\sum_{i=1}^{k-1} v_{i,j}w_{i,j}\overline{v_{i,k}w_{i,k}}\right) +  v_{k,j}w_{k,j}\overline{v_{k,k}w_{k,k}} = c_{j,k} +  v_{k,j}w_{k,j}\overline{v_{k,k}w_{k,k}},\\
    \end{split}\end{align}
    where the second equality again follows from the fact that $v_{i,k} = 0$ whenever $i > k$. Let's now compute the final term above (i.e., the term that we pulled out of the sum) by plugging in the formulas~\eqref{eq:vk_wk_compact} for $v_{k,j}$, $w_{k,j}$, and $w_{k,k}$:
    \begin{align*}
    	v_{k,j}w_{k,j}\overline{v_{k,k}w_{k,k}} & = \frac{x_{j,k} - c_{j,k}}{\sqrt{y_{k,j} - d_{k,j}}}\left(\frac{\sqrt{y_{k,j} - d_{k,j}}}{{v_{k,k}}}\right)\big(\overline{v_{k,k}}\big)\left(\frac{\sqrt{y_{k,k} - d_{k,k}}}{\overline{v_{k,k}}}\right) = \frac{\sqrt{y_{k,k} - d_{k,k}}}{v_{k,k}}\big(x_{j,k} - c_{j,k}\big),
    \end{align*}
    where the second equality just follows from cancelling terms where possible. If we now plug in the formula~\eqref{eq:vk_wk_compact} for $v_{k,k}$ then we see that this quantity equals
    \begin{align*}
    	v_{k,j}w_{k,j}\overline{v_{k,k}w_{k,k}} & = (y_{k,k} - d_{k,k})\left(\frac{x_{j,k} - c_{j,k}}{x_{k,k} - c_{k,k}}\right) = x_{j,k} - c_{j,k},
    \end{align*}
    where the final equality follows from the facts that $y_{k,k} = x_{k,k}$ and $d_{k,k} = c_{k,k}$. Finally, plugging this equation into Equation~\eqref{eq:start_of_X_suff_decomp} gives
    \[
        \left[\sum_{i=1}^n \mathbf{(v_i\odot w_i)}\mathbf{(v_i\odot w_i)}^*\right]_{j,k} = c_{j,k} + v_{k,j}w_{k,j}\overline{v_{k,k}w_{k,k}} = c_{j,k} + (x_{j,k} - c_{j,k}) = x_{j,k}.
    \]
    It follows that this this PCP decomposition gives the correct $X$ matrix as well, completing the proof.%
\end{proof}

We note that the decomposition of the above theorem can be thought of as constructing $X$ and $Y$ one row at a time. The choice of $\mathbf{v_1}$ and $\mathbf{w_1}$ give the correct first row of $X$ and $Y$, then $\mathbf{v_2}$ and $\mathbf{w_2}$ correct the second row of $X$ and $Y$ (without affecting the entries in their first row), and so on. For this reason, there are cases where the above theorem is not able to directly find a PCP decomposition of a pair $(X,Y)$, but it is able to find a PCP decomposition for $(PXP^*,PYP^*)$ for some permutation matrix $P$.

\begin{example}\label{exam:weird_qutrit}
    Consider the pair of matrices
    \[
        X = \begin{bmatrix}
            2 & 1 & -1 \\
            1 & 8 & 1 \\
            -1 & 1 & 4
        \end{bmatrix} \quad \text{and} \quad Y = \begin{bmatrix}
            2 & 1 & 3 \\
            2 & 8 & 1 \\
            1 & 2 & 4
        \end{bmatrix}.
    \]
    
    If we try to construct a PCP decomposition of $(X,Y)$ via Theorem~\ref{thm:ccp_sufficient_large}, we get stuck when computing $v_{2,3} = 3/\sqrt{-2}$, since it contains a negative number underneath a square root so the theorem does not apply. To get around this problem, we can conjugate $X$ and $Y$ by the permutation matrix
    \[
        P = \begin{bmatrix}
            0 & 1 & 0 \\
            0 & 0 & 1 \\
            1 & 0 & 0
        \end{bmatrix}.
    \]
    Then the theorem gives us the following PCP decomposition of $(PXP^*,PYP^*)$:
    \begin{align*}
        \mathbf{v_1} & = (\sqrt{8}, 1, 1/\sqrt{2})^T, \quad & \mathbf{w_1} & = (1, 1/\sqrt{8}, 1/2)^T, \\
        \mathbf{v_2} & = (0,\sqrt{31/8},-3\sqrt{3}/4)^T, \quad & \mathbf{w_2} & = (\sqrt{8/31},1,1/\sqrt{6/31})^T, \\
        \mathbf{v_3} & = (0,0,4\sqrt{3/31})^T, \quad & \mathbf{w_3} & = (1/(2\sqrt{6}),\sqrt{155/192},1)^T.
    \end{align*}
    Simply permuting the entries of these vectors according to $P^{-1}$ then gives us the following PCP decomposition of $(X,Y)$:
    \begin{align*}
        \mathbf{v_1} & = (1/\sqrt{2}, \sqrt{8}, 1)^T, \quad & \mathbf{w_1} & = (1/2, 1, 1/\sqrt{8})^T, \\
        \mathbf{v_2} & = (-3\sqrt{3}/4,0,\sqrt{31/8})^T, \quad & \mathbf{w_2} & = (1/\sqrt{6/31},\sqrt{8/31},1)^T, \\
        \mathbf{v_3} & = (4\sqrt{3/31},0,0)^T, \quad & \mathbf{w_3} & = (1,1/(2\sqrt{6}),\sqrt{155/192})^T.
    \end{align*}
\end{example}

One more sufficient condition for showing that matrices are PCP can be motivated by recalling that every entrywise non-negative diagonally-dominant matrix (i.e., every matrix $X \in M_n(\mathbb{R})$ satisfying $x_{j,j} \geq \sum_{i \neq j} x_{i,j}$ for all $j$ and $x_{i,j} \geq 0$ for all $i,j$) is completely positive \cite{Kay87}. As a generalization of this fact, it was shown in \cite{DJL94} that if $X$ is doubly non-negative and its \emph{comparison matrix} $M(X)$ is positive semidefinite, then $X$ is completely positive, where the comparison matrix $M(X)$ is constructed by replacing the off-diagonal entries of $X$ by the negative of their absolute value:
\[
    M(X) \defeq \begin{bmatrix}
        |x_{1,1}| & -|x_{1,2}| & \cdots & -|x_{1,n}| \\
        -|x_{2,1}| & |x_{2,2}| & \cdots & -|x_{2,n}| \\
        \vdots & \vdots & \ddots & \vdots \\
        -|x_{n,1}| & -|x_{n,2}| & \cdots & |x_{n,n}|
    \end{bmatrix}.
\]
The following theorem establishes the natural generalization of these facts to pairwise completely positive matrices.

\begin{theorem}\label{thm:m_matrix_ccp}
    Suppose $X,Y \in M_n(\mathbb{C})$ satisfy conditions (a)--(d) of Theorem~\ref{thm:ccp_properties}. If $M(X)$ is positive semidefinite then $(X,Y)$ is pairwise completely positive.
\end{theorem}

\begin{proof}
    We start by showing that if conditions (a)--(d) of Theorem~\ref{thm:ccp_properties} hold and $X$ is diagonally dominant then $(X,Y)$ is PCP.
    
    First, we note that we can assume without loss of generality that $x_{i,i}= \sum_{j \neq i} |x_{i,j}|$ for all $1 \leq i \leq n$ and $y_{i,j}y_{j,i}=|x_{i,j}|^{2}$ for all $1 \leq i \neq j \leq n$, since the more general case where the quantities on the left are larger than the quantities on the right then follows from Lemma~\ref{lem:ccp_increate_y_offdiag}.
    
    Using an approach motivated by \cite{Kay87}, we construct matrices $V,W, \in M_{n,m}(\mathbb{C})$ with $m = n(n-1)/2$ such that $X = (V\odot W)(V\odot W)^{*}$ and $Y = (V\odot \overline{V})(W\odot \overline{W})^{*}$. We index the $n$ rows of $V$ and $W$ in the usual way (by a number $1 \leq j \leq n$), but we index their $n(n-1)/2$ columns by sets $\{k,\ell\}$ for which $k \neq \ell$ (we will think of these columns as corresponding to the $n(n-1)/2$ entries in the strictly upper-triangular portion of $X$ and $Y$).
    
    We define $V$ and $W$ entrywise as follows (we use the convention that $k < \ell$ throughout this definition):
    \begin{align*}
        v_{j,\{k,\ell\}} = \begin{cases}
            \operatorname{sign}(x_{k,\ell})y_{k,\ell}^{1/4} & \text{if} \ j = k \\
            y_{\ell,k}^{1/4} & \text{if} \ j = \ell \\
            0 & \text{otherwise}
        \end{cases} \qquad \text{and} \qquad w_{j,\{k,\ell\}} = \begin{cases}
            y_{\ell,k}^{1/4} & \text{if} \ j = k \\
            y_{k,\ell}^{1/4} & \text{if} \ j = \ell \\
            0 & \text{otherwise},
        \end{cases}
    \end{align*}
    where $\operatorname{sign}(x_{k,\ell})$ is the (complex) sign of $x_{k,\ell}$: the number on the unit circle in the complex plane with the property that $x_{k,\ell} = \operatorname{sign}(x_{k,\ell})|x_{k,\ell}|$. Then
    \begin{align*}
        \big[(V \odot \overline{V})(W \odot \overline{W})^*\big]_{i,j}=\sum_{k < \ell}|v_{i,\{k,\ell\}}|^{2}|w_{j,\{k,\ell\}}|^{2}.
    \end{align*}
    When $i \neq j$, the only nonzero term in this sum arises when $\{k,\ell\} = \{i,j\}$, so it simplifies considerably to
    \begin{align*}
        \big[(V \odot \overline{V})(W \odot \overline{W})^*\big]_{i,j} & = |v_{i,\{i,j\}}|^{2}|w_{j,\{i,j\}}|^{2} = (y_{i,j}^{1/4})^2 (y_{i,j}^{1/4})^2 = y_{i,j},
    \end{align*}
    as desired. Similarly, if $i < j$ then the $(i,j)$-entry of $X$ is given by
    \begin{align*}
        \big[(V \odot W)(V \odot W)^* \big]_{i,j} & = v_{i,\{i,j\}}w_{i,\{i,j\}}\overline{v_{j,\{i,j\}}}\overline{w_{j,\{i,j\}}} = \operatorname{sign}(x_{i,j})y_{i,j}^{1/4} \times y_{j,i}^{1/4} \times y_{j,i}^{1/4} \times y_{i,j}^{1/4} = x_{i,j},
    \end{align*}
    with the computation when $i > j$ being analogous.
    
    Finally, we show the given factorization yields the correct diagonal entries $y_{i,i} = x_{i,i} = \sum_{k \neq i}|x_{i,k}|$:
    \begin{align*}
        \big[(V \odot \overline{V})(W \odot \overline{W})^* \big]_{i,i} = \sum_{k < \ell}|v_{i,\{k,\ell\}}|^{2}|w_{i,\{k,\ell\}}|^{2} = \sum_{k \neq i} y_{i,k}^{1/2}y_{k,i}^{1/2} = \sum_{k\neq i} |x_{i,k}| = x_{k,k},
    \end{align*}
    which completes the proof that if conditions (a)--(d) of Theorem~\ref{thm:ccp_properties} hold and $X$ is diagonally dominant then $(X,Y)$ is PCP. 
    
    Next, we show that if $M = M(X)$ is positive semidefinite then there exists an entrywise non-negative diagonal matrix $D$ such that $DXD$ is diagonally dominant. To this end, notice that we can write $M = \alpha I - P$, where $P \geq 0$ is a real symmetric entrywise non-negative matrix, and positive semidefiniteness of $M$ tells us that $\alpha \geq \lambda_{\textup{max}}(P)$, where $\lambda_{\textup{max}}(P)$ is the largest eigenvalue of $P$. Without loss of generality, we can assume that $P$ is irreducible and thus has a Perron eigenvector $\mathbf{x}$ with strictly positive entries (otherwise we can either use a continuity argument or the fact that every entrywise non-negative $P$ can, up to permutation similarity, be written in a block upper triangular form with irreducible diagonal blocks).
    
    If we let $D = \operatorname{diag}(\mathbf{x})$ then
    \begin{align*}
        DMD\mathbf{1} = DM\mathbf{x} = D(\alpha I - P)\mathbf{x} = D(\alpha \mathbf{x} - \lambda_{\textup{max}}(P)\mathbf{x}) = D(\alpha - \lambda_{\textup{max}}(P))\mathbf{x} \geq 0,
    \end{align*}
    since we already noted that $\alpha \geq \lambda_{\textup{max}}(P)$. Since $DMD\mathbf{1} \geq 0$ we know that $DMD$ is diagonally dominant, and since $DXD$ has the same entrywise absolute values as $DMD$, we conclude that $DXD$ is diagonally dominant as well.
    
    It is straightforward to check that if $(X,Y)$ satisfies conditions (a)--(d) of Theorem~\ref{thm:ccp_properties} then so does $(DXD,DYD)$, so we have just shown that $(DXD,DYD)$ is PCP. It follows that \[(D^{-1}DXDD^{-1},D^{-1}DYDD^{-1}) = (X,Y)\]is PCP as well, which completes the proof.
\end{proof}

It is perhaps worth noting that the above theorem does \emph{not} follow directly from combining the known result that diagonal dominance implies complete positivity with Theorem~\ref{thm:cp_is_ccp}, since the result for completely positive matrices only applies to matrices $X$ with non-negative entries. Our result, however, applies even if $X$ has negative (or complex) entries.

For example, the above theorem provides another way of seeing that the pair $(X,Y)$ from Example~\ref{exam:werner_ccp} is PCP when $a = n$ and $b = -1$, since the matrix $X$ in this pair is diagonally dominant and $Y \geq X$. In fact, it even shows that the pair $(X,Y)$ with $X = nI_n - J_n$ and $Y = (n-2)I_n + J_n$ is PCP for the exact same reason (and this is a much stronger statement, since this $Y$ has much smaller off-diagonal entries). Similarly, this theorem also shows that the pair $(X,Y)$ from Example~\ref{exam:weird_qutrit} is PCP and can be used to construct another PCP decomposition of it.

\begin{example}
Consider the pair of matrices
\[
    X = \begin{bmatrix}
        2 & 1 & -1 \\
        1 & 3 & 2i \\
        -1 & -2i & 3
    \end{bmatrix} \quad \text{and} \quad Y = \begin{bmatrix}
        2 & 1 & 2 \\
        1 & 3 & 4 \\
        1/2 & 1 & 3
    \end{bmatrix}.
\]
Since this pair satisfies conditions~(a)--(d) of Theorem~\ref{thm:ccp_properties} and $X$ is diagonally dominant, Theorem~\ref{thm:m_matrix_ccp} tells us that $(X,Y)$ is pairwise completely positive. To construct an explicit PCP decomposition of it, we compute the matrices $V,W\in M_{n,m}(\mathbb{C})$ by following along through the proof of that theorem. Since $n=3$, these matrices have three rows indexed by $1 \leq j \leq 3$, and $m = n(n-1)/2=3$ columns indexed by the sets $\{k,\ell\}$ with $1 \leq k < \ell\leq 3$. The only non-zero entries of $V$ and $W$ are the ones for which $j$ equals either $k$ or $\ell$, so we know there are two entries to be computed for each of the three pairs, resulting in six non-zero entries in each of these matrices. For example, if $\{k,\ell\}=\{1,2\}$ then we have
\begin{align*}
    v_{1,\{1,2\}}&= \operatorname{sign}(x_{1,2})y_{1,2}^{1/4}=1^{1/4} = 1&w_{1,\{1,2\}}&=y_{2,1}^{1/4}=1^{1/4}=1\\
    v_{2,\{1,2\}}&=y_{2,1}^{1/4}=1^{1/4}=1 \quad &w_{2,\{1,2\}}&=y_{1,2}^{1/4}=1^{1/4}=1.
\end{align*}
Similar computations show that
\begin{align*}
     v_{1,\{1,3\}}&= \operatorname{sign}(x_{1,3})y_{1,3}^{1/4}=-2^{1/4}&w_{1,\{1,3\}}&=y_{3,1}^{1/4}=(1/2)^{1/4}=2^{-1/4} \\
     v_{3,\{1,3\}}&=y_{3,1}^{1/4}=(1/2)^{1/4}=2^{-1/4} \quad & w_{3,\{1,3\}}&=y_{1,3}^{1/4}=2^{1/4}, \\
	 v_{2,\{2,3\}}&= \operatorname{sign}(x_{2,3})y_{2,3}^{1/4}=i\sqrt{2} &w_{2,\{2,3\}}&=y_{3,2}^{1/4}=1^{1/4}=1 \\
	 v_{3,\{2,3\}}&=y_{3,2}^{1/4}=1^{1/4}=1 \quad & w_{3,\{2,3\}}&=y_{2,3}^{1/4}=4^{1/4}=\sqrt{2}.
\end{align*}
With all of the non-zero entries computed, we can now construct the matrices $V$ and $W$ that make up the PCP decomposition of $(X,Y)$:
\[
    V = \begin{bmatrix}
        v_{1,\{1,2\}} & v_{1,\{1,3\}} & v_{1,\{2,3\}} \\
        v_{2,\{1,2\}} & v_{2,\{1,3\}} & v_{2,\{2,3\}} \\
        v_{3,\{1,2\}} & v_{3,\{1,3\}} & v_{3,\{2,3\}} \\
    \end{bmatrix} = \begin{bmatrix}
        1 & -{2}^{1/4} & 0 \\
        1 & 0 & i\sqrt{2}\\
        0 & 2^{-1/4} & 1 
    \end{bmatrix} \quad \text{and} \quad W = \begin{bmatrix}
        1 & 2^{-1/4} & 0 \\
        1 & 0 & 1 \\
        0 &2^{1/4} & \sqrt{2}
    \end{bmatrix}.
\]
\end{example}

\section{Connection with Quantum Entanglement}\label{sec:quantum_sep}

A \emph{quantum state} is a positive semidefinite matrix $\rho \in M_n(\mathbb{C})$ with trace $1$. From now on, whenever we use lowercase Greek letters like $\rho$ or $\sigma$, we are implicitly assuming that it is a quantum state. One of the central questions in quantum information theory asks how to determine whether or not a state $\rho \in M_n(\mathbb{C}) \otimes M_n(\mathbb{C})$ can be written in the form
\begin{align}\label{eq:rho_separable}
    \rho = \sum_{k=1}^m \mathbf{v}_{\mathbf{k}}\mathbf{v}_{\mathbf{k}}^* \otimes \mathbf{w}_k\mathbf{w}_k^*
\end{align}
for some families of vectors $\{\mathbf{v}_k\},\{\mathbf{w}_k\} \subseteq \mathbb{C}^n$. States of this form are called \emph{separable} \cite{Wer89}, whereas states that cannot be written in this form are called \emph{entangled}. Note that $\rho$ being scaled to have trace $1$ is not important from a mathematical perspective here---the same question could be asked of any positive semidefinite matrix, so from now on we do not place any restriction on the trace of quantum states.

The problem of determining whether a state is separable or entangled is NP-hard \cite{Gha10,Gur03}, but there are many necessary or sufficient conditions that can be used to show that certain specific states are separable or entangled. The most well-known such test is the \emph{positive partial transpose (PPT)} criterion \cite{Per96,HHH96}, which says that if $\rho$ is separable then $(id \otimes T)(\rho)$ is positive semidefinite, where $id, T : M_n(\mathbb{C}) \rightarrow M_n(\mathbb{C})$ are the identity and transposition maps, respectively. Also of interest for us is the \emph{realignment criterion}, which says that if $\rho$ is separable and $R$ is the linear map on $M_n(\mathbb{C}) \otimes M_n(\mathbb{C})$ defined by $R(\mathbf{e_i}\mathbf{e}_{\mathbf{j}}^* \otimes \mathbf{e_k}\mathbf{e}_{\mathbf{\ell}}^*) = \mathbf{e_i}\mathbf{e}_{\mathbf{k}}^* \otimes \mathbf{e_j}\mathbf{e}_{\mathbf{\ell}}^*$ then $\|R(\rho)\|_{\textup{tr}} \leq \tr(\rho)$ \cite{CW03,Rud00}. For a more thorough treatment of the problem of showing that a state is separable or entangled, see the review articles \cite{GT09,HHH09}.

One of the main results of \cite{Yu16,TAQLS17} says that, for a family of quantum states called \emph{mixed Dicke states}, separability of the state is determined exactly by complete positivity of an associated matrix. The main result of this section shows that an analogous result holds for pairwise completely positive matrices and a more general family of quantum states. That is, there exists a family of quantum states that contain the mixed Dicke states as a subset with the property that they are separable if and only if an associated pair of matrices is pairwise completely positive. To begin making this more precise, we introduce the following family of quantum states:

\begin{definition}
    A mixed state $\rho \in M_n(\mathbb{C}) \otimes M_n(\mathbb{C})$ is called a \emph{conjugate local diagonal unitary invariant (CLDUI) state} if
    \begin{align*}
        (U \otimes \overline{U})\rho(U \otimes \overline{U})^* = \rho
    \end{align*}
    for all diagonal unitary matrices $U \in M_n(\mathbb{C})$.
\end{definition}

We claim that CLDUI states are exactly those that can be written in the form
\begin{align}\label{eq:LDUI_state_form}
    \rho = \sum_{i,j=1}^n x_{i,j} \mathbf{e_i}\mathbf{e}_{\mathbf{j}}^* \otimes \mathbf{e_i}\mathbf{e}_{\mathbf{j}}^* + \sum_{i \neq j=1}^n y_{i,j} \mathbf{e_i}\mathbf{e}_{\mathbf{i}}^* \otimes \mathbf{e_j}\mathbf{e}_{\mathbf{j}}^*.
\end{align}
If we collect the coefficients $\{x_{i,j}\}$ and $\{y_{i,j}\}$ into matrices $X$ and $Y$ in the usual way (defining $y_{i,i} = x_{i,i}$ for all $i$ so that $Y$ has diagonal entries) then this gives us a correspondence between CLDUI states $\rho$ and matrix pairs $(X,Y)$ for which $X \succeq 0$ (since $X$ is a submatrix of $\rho$), $Y \geq 0$ (since $Y$ consists of diagonal entries of $\rho$), and the diagonal entries of $X$ and $Y$ coincide. Given a pair of matrices $(X,Y)$ satisfying these three properties (which are exactly properties (a)--(c) of Theorem~\ref{thm:ccp_properties}), we use $\rho_{X,Y}$ to denote the associated CLDUI state~\eqref{eq:LDUI_state_form} by $\rho_{X,Y}$.

To verify that CLDUI states do indeed have the claimed form, simply observe that that if we use $[\rho]_{i,j,k,\ell} = \rho_{i,j,k,\ell}$ to denote the coefficient of the basis matrix $\mathbf{e_i}\mathbf{e}_{\mathbf{j}}^* \otimes \mathbf{e_k}\mathbf{e}_{\mathbf{\ell}}^*$ in $\rho$, then
\begin{align*}
    [(U \otimes \overline{U})\rho(U \otimes \overline{U})^*]_{i,j,k,\ell} = u_i \overline{u_j} \overline{u_k}u_\ell\rho_{i,j,k,\ell},
\end{align*}
which equals $\rho_{i,j,k,\ell}$ for all diagonal $U$ if and only if $(i,\ell) = (j,k)$, $(i,\ell) = (k,j)$, or $\rho_{i,j,k,\ell} = 0$.

For example, in the $(3 \otimes 3)$-dimensional case, every CLDUI state $\rho$ can be written in the standard basis in the following form, where we use $\cdot$ to denote entries equal to $0$:
\begin{align*}
	\rho = \left[\begin{array}{ccc|ccc|ccc}
        x_{1,1} & \cdot & \cdot & \cdot & x_{1,2} & \cdot & \cdot & \cdot & x_{1,3} \\
        \cdot & y_{1,2} & \cdot & \cdot & \cdot & \cdot & \cdot & \cdot & \cdot \\
        \cdot & \cdot & y_{1,3} & \cdot & \cdot & \cdot & \cdot & \cdot & \cdot \\\hline
        \cdot & \cdot & \cdot & y_{2,1} & \cdot & \cdot & \cdot & \cdot & \cdot \\
        x_{2,1} & \cdot & \cdot & \cdot & x_{2,2} & \cdot & \cdot & \cdot & x_{2,3} \\
        \cdot & \cdot & \cdot & \cdot & \cdot & y_{2,3} & \cdot & \cdot & \cdot \\\hline
        \cdot & \cdot & \cdot & \cdot & \cdot & \cdot & y_{3,1} & \cdot & \cdot \\
        \cdot & \cdot & \cdot & \cdot & \cdot & \cdot & \cdot & y_{3,2} & \cdot \\
        x_{3,1} & \cdot & \cdot & \cdot & x_{3,2} & \cdot & \cdot & \cdot & x_{3,3}
    \end{array}\right].
\end{align*}

CLDUI states include several other well-known families of quantum states as special cases. For example, \emph{isotropic states} \cite{HH99} which are those that are invariant under \emph{every} (not necessarily diagonal) unitary of the form $U \otimes \overline{U}$, and these states are exactly the ones for which $X = aI_n + bJ_n$ and $Y = bI_n + aJ_n$ for some $a,b \in \mathbb{R}$. Also, every mixed Dicke state \cite{Yu16,TAQLS17} is the partial transpose of a CLDUI matrix (these are exactly the states that are obtained when $X = Y$).

The following result shows that separability of a CLDUI state $\rho_{X,Y}$ is determined exactly by whether or not the associated pair of coefficient matrices $(X,Y)$ is pairwise completely positive, as well as how the other necessary conditions of Theorem~\ref{thm:ccp_properties} relate to known separability criteria.

\begin{theorem}\label{thm:ccp_separable}
    Suppose $X,Y \in M_n(\mathbb{C})$. Then the pair of matrices $(X,Y)$ has the following relationship with properties of the CLDUI state $\rho_{X,Y}$:
    \begin{enumerate}
        \item[a)] $\rho_{X,Y}$ is a positive semidefinite if and only if $X$ is positive semidefinite, $Y$ is entrywise non-negative, and $X$ and $Y$ have the same diagonal entries.
        
        \item[b)] $\rho_{X,Y}$ is furthermore a valid quantum state (i.e., $\tr(\rho_{X,Y}) = 1$) if and only if $\|Y\|_1 = 1$.
        
        \item[c)] $\rho_{X,Y}$ is separable if and only if $(X,Y)$ is pairwise completely positive.
        
        \item[d)] $\rho_{X,Y}$ has positive partial transpose if and only if $(X,Y)$ satisfies property~(d) of Theorem~\ref{thm:ccp_properties}.
        
        \item[e)] $\rho_{X,Y}$ satisfies the realignment criterion (i.e., $\|R(\rho_{X,Y})\|_{\textup{tr}} \leq \tr(\rho)$) if and only if $(X,Y)$ satisfies property~(e) of Theorem~\ref{thm:ccp_properties}.
    \end{enumerate}
\end{theorem}

\begin{proof}
    We already proved the ``only if'' direction of part~(a) above. For the converse, we just note that if $X$ and $Y$ have the three properties described then $\rho$ must be positive semidefinite since up to permutation similarity it can be written as the direct sum $X \oplus \left(\bigoplus_{i\neq j=1}^n y_{i,j}\right)$, and each term in that direct sum is positive semidefinite.
    
    Part~(b) follows simply from noting that $Y \geq 0$ and that it consists of the diagonal entries of $\rho_{X,Y}$.
    
    For part~(c), first note that for every rank-$1$ separable state $\mathbf{v}\mathbf{v}^* \otimes \mathbf{w}\mathbf{w}^*$, it is the case that integrating over diagonal unitaries with respect to Haar measure gives
    \begin{align*}
    	\int_U (U \otimes \overline{U}) \big(\mathbf{v}\mathbf{v}^* \otimes \mathbf{w}\mathbf{w}^*\big) (U \otimes \overline{U})^* \, dU & = \sum_{i,j=1}^n v_i\overline{v_j}\overline{w_i}w_j\mathbf{e_i}\mathbf{e}_{\mathbf{j}}^* \otimes \mathbf{e_i}\mathbf{e}_{\mathbf{j}}^* + \sum_{i \neq j=1}^n |v_i|^2|w_j|^2\mathbf{e_i}\mathbf{e}_{\mathbf{i}}^* \otimes \mathbf{e_j}\mathbf{e}_{\mathbf{j}}^*.
    \end{align*}
    Well if $\rho_{X,Y}$ is a CLDUI state then $\rho_{X,Y} = \int_U (U \otimes \overline{U}) \rho_{X,Y} (U \otimes \overline{U})^* \, dU$, so it follows that $\rho_{X,Y}$ is separable (i.e., $\rho_{X,Y} = \sum_{k} \mathbf{v}_{\mathbf{k}}\mathbf{v}_{\mathbf{k}}^* \otimes \mathbf{w}_k\mathbf{w}_k^*$) if and only if the pair of matrices $(X,Y)$ satisfy
    \[
        x_{i,j} = \sum_k v_{k,i}\overline{v_{k,j}}\overline{w_{k,i}}w_{k,j} \quad \text{and} \quad y_{i,j} = \sum_k |v_{k,i}|^2|w_{k,j}|^2 \quad \text{for all} \quad 1 \leq i,j \leq n.
    \]
    In other words, $\rho_{X,Y}$ is separable if and only if
    \[
    	X = \sum_k \mathbf{(v_k\odot \overline{w_k})}\mathbf{(v_k\odot \overline{w_k})}^* \qquad \text{and} \qquad Y = \sum_k \mathbf{(v_k\odot \overline{v_k})}\mathbf{(w_k\odot\overline{w_k})}^*,
    \]
    which is a PCP decomposition of $(X,Y)$ (to get this decomposition in the form~\eqref{eq:ccp_defn_sum}, just replace each $\mathbf{w_k}$ by $\overline{\mathbf{w_k}}$).
    
    For part~(d), we notice that if $\rho_{X,Y}$ is CLDUI then up to permutation similarity $(id \otimes T)(\rho_{X,Y})$ can be written as the direct sum of $1 \times 1$ and $2 \times 2$ matrices:
    \[
        (id \otimes T)(\rho_{X,Y}) = \left(\bigoplus_{i=1}^n x_{i,i}\right) \oplus \left(\bigoplus_{i<j} \begin{bmatrix}
            y_{i,j} & x_{i,j} \\
            x_{j,i} & y_{j,i}
        \end{bmatrix}\right).
    \]
    Since $y_{i,j}, y_{j,i} \geq 0$ for all $i,j$ and $x_{i,i} \geq 0$, we conclude that $(id \otimes T)(\rho_{X,Y})$ is positive semidefinite if and only if the determinant of each of the $2 \times 2$ matrices is non-negative---i.e., $y_{i,j}y_{j,i} \geq |x_{i,j}|^2$.
    
    Finally, for part~(e) we note that $Y$ contains exactly the diagonal entries of $\rho_{X,Y}$, so $\tr(\rho_{X,Y}) = \|Y\|_1$. Also,
    \[
        R(\rho_{X,Y}) = \sum_{i,j=1}^n x_{i,j} \mathbf{e_i}\mathbf{e}_{\mathbf{i}}^* \otimes \mathbf{e_j}\mathbf{e}_{\mathbf{j}}^* + \sum_{i \neq j=1}^n y_{i,j} \mathbf{e_i}\mathbf{e}_{\mathbf{j}}^* \otimes \mathbf{e_i}\mathbf{e}_{\mathbf{j}}^*,
    \]
    and since this matrix is block diagonal its trace norm simplifies to $\|R(\rho_{X,Y})\|_{\textup{tr}} = (\|X\|_1 - \tr(X)) + \|Y\|_{\textup{tr}}$. By using the fact that $\tr(X) = \|X\|_{\textup{tr}}$ (since $X$ is positive semidefinite) and rearranging, we see that the realignment criterion $\|R(\rho_{X,Y})\|_{\textup{tr}} \leq \tr(\rho)$ is equivalent to $\|X\|_1 - \|X\|_{\textup{tr}} \leq \|Y\|_1 - \|Y\|_{\textup{tr}}$.
\end{proof}

\begin{example}\label{exam:3_dim_sep_eg}
	The pair $(X,Y)$ from Example~\ref{exam:3dim_eg} given by
	\begin{align*}
	X = \begin{bmatrix}
    	1 & 1 & 1 \\
    	1 & 1 & 1 \\
    	1 & 1 & 1
	\end{bmatrix} \quad \text{and} \quad Y = \begin{bmatrix}
		1 & a & 1/a \\
		1/a & 1 & a \\
		a & 1/a & 1
	\end{bmatrix}
	\end{align*}
	corresponds to the CLDUI state
	\[
        \rho_{X,Y} = \left[\begin{array}{ccc|ccc|ccc}
            1 & \cdot & \cdot & \cdot & 1 & \cdot & \cdot & \cdot & 1 \\
            \cdot & a & \cdot & \cdot & \cdot & \cdot & \cdot & \cdot & \cdot \\
            \cdot & \cdot & 1/a & \cdot & \cdot & \cdot & \cdot & \cdot & \cdot \\\hline
            \cdot & \cdot & \cdot & 1/a & \cdot & \cdot & \cdot & \cdot & \cdot \\
            1 & \cdot & \cdot & \cdot & 1 & \cdot & \cdot & \cdot & 1 \\
            \cdot & \cdot & \cdot & \cdot & \cdot & a & \cdot & \cdot & \cdot \\\hline
            \cdot & \cdot & \cdot & \cdot & \cdot & \cdot & a & \cdot & \cdot \\
            \cdot & \cdot & \cdot & \cdot & \cdot & \cdot & \cdot & 1/a & \cdot \\
            1 & \cdot & \cdot & \cdot & 1 & \cdot & \cdot & \cdot & 1
        \end{array}\right].
	\]
	Indeed, this state has positive partial transpose for all $a > 0$ (which corresponds to the fact that $(X,Y)$ satisfies condition~(d) of Theorem~\ref{thm:ccp_properties} for all $a > 0$). However, it is straightforward to show that it only satisfies the realignment criterion when $a = 1$ (which corresponds to the fact that $(X,Y)$ satisfies condition~(e) of Theorem~\ref{thm:ccp_properties} if and only if $a = 1$).
	
	The fact that $(X,Y)$ is PCP when $a = 1$ tells us that $\rho_{X,Y}$ is separable when $a = 1$. Again, this fact is known, but it is not straightforward to see (i.e., there is no ``obvious'' separable decomposition of $\rho_{X,Y}$). We elaborate on this point of why PCP decomposition can be simple without there being a simple separable decomposition in Section~\ref{sec:separable_length}.
\end{example}

\begin{example}\label{exam:werner_ccp_revisited}
    The pair $(X,Y)$ from Example~\ref{exam:werner_ccp} given by
    \begin{align*}
        X = aI_n + bJ_n \quad \text{and} \quad Y = bI_n + aJ_n
    \end{align*}
    corresponds to the well-known \emph{isotropic states}---states of the form
    \[
        \rho_{X,Y} = a(I_n \otimes I_n) + b\left(\sum_{i=1}^n \mathbf{e_i} \otimes \mathbf{e_i}\right)\left(\sum_{i=1}^n \mathbf{e_i} \otimes \mathbf{e_i}\right)^*.
    \]
    The fact that the pairs $(X,Y)$ satisfy condition~(d) of Theorem~\ref{thm:ccp_properties} if and only if they are PCP if and only if $a \geq 0$ and $-a/n \leq b \leq a$ corresponds to the well-known fact that isotropic states have positive partial transpose if and only if they are separable if and only if $a \geq 0$ and $-a/n \leq b \leq a$.
\end{example}

Although the conditions from Theorem~\ref{thm:ccp_properties} all correspond to well-known separability criteria, to the best of our knowledge the tests from Theorems~\ref{thm:ccp_n2} and~\ref{thm:m_matrix_ccp} provide truly new methods of showing that CLDUI states are separable. Of particular note is the following result:

\begin{corollary}\label{cor:ldui_state_m_matrix}
    Suppose $\rho \in M_n(\mathbb{C}) \otimes M_n(\mathbb{C})$ is a CLDUI state. If $\rho$ has positive partial transpose and $M(\rho)$ is positive semidefinite then $\rho$ is separable.
\end{corollary}

\begin{proof}
    This follows immediately from combining Theorems~\ref{thm:m_matrix_ccp} and~\ref{thm:ccp_separable}.
\end{proof}

We will make use of this corollary shortly to answer a question about absolutely separable states. We leave the question of how to generalize Corollary~\ref{cor:ldui_state_m_matrix} to arbitrary states (or even \emph{if} it can be generalized to arbitrary states in any meaningful way) as an open question.

\subsection{The Length of Decompositions}\label{sec:separable_length}

We now clarify that, even though Theorem~\ref{thm:ccp_separable} shows that separable CLDUI states and PCP matrices are in one-to-one correspondence, it can be significantly easier to come up with a PCP decomposition of a pair $(X,Y)$ than to come up with a separable decomposition of the corresponding CLDUI state $\rho_{X,Y}$. To make this observation a bit more explicit, we define the \emph{length} of a PCP pair $(X,Y)$, denoted by $\ell(X,Y)$, to be the least integer $m$ in the PCP decomposition~\eqref{eq:ccp_defn_sum} (i.e., the smallest number of terms in the sum in a PCP decomposition). The \emph{length} of a separable state $\rho$, denoted by $\ell(\rho)$, is similarly defined to be the least integer $m$ in the separable decomposition~\eqref{eq:rho_separable}.

It turns out that the length of a PCP pair $(X,Y)$ may be significantly smaller than the length of the corresponding CLDUI state $\rho_{X,Y}$. For example, we noted in Example~\ref{exam:3dim_eg} that the pair $(X,Y)$ with
\begin{align}\label{eq:pcp_all_ones}
    X = Y = \begin{bmatrix}
        1 & 1 & 1 \\
        1 & 1 & 1 \\
        1 & 1 & 1
    \end{bmatrix}
\end{align}
is PCP since it has the trivial decomposition $\mathbf{v_1} = \mathbf{w_1} = (1,1,1)^T$. In particular, this means that $(X,Y)$ has length~$1$. On the other hand, the corresponding quantum state $\rho_{X,Y} \in M_3(\mathbb{C}) \otimes M_3(\mathbb{C})$ is
\[
    \rho_{X,Y} = \left[\begin{array}{ccc|ccc|ccc}
        1 & \cdot & \cdot & \cdot & 1 & \cdot & \cdot & \cdot & 1 \\
        \cdot & 1 & \cdot & \cdot & \cdot & \cdot & \cdot & \cdot & \cdot \\
        \cdot & \cdot & 1 & \cdot & \cdot & \cdot & \cdot & \cdot & \cdot \\\hline
        \cdot & \cdot & \cdot & 1 & \cdot & \cdot & \cdot & \cdot & \cdot \\
        1 & \cdot & \cdot & \cdot & 1 & \cdot & \cdot & \cdot & 1 \\
        \cdot & \cdot & \cdot & \cdot & \cdot & 1 & \cdot & \cdot & \cdot \\\hline
        \cdot & \cdot & \cdot & \cdot & \cdot & \cdot & 1 & \cdot & \cdot \\
        \cdot & \cdot & \cdot & \cdot & \cdot & \cdot & \cdot & 1 & \cdot \\
        1 & \cdot & \cdot & \cdot & 1 & \cdot & \cdot & \cdot & 1
    \end{array}\right],
\]
which is indeed separable, but it is much more difficult to ``directly'' see that this is the case. Indeed, it is straightforward to show that $\ell(\rho) \geq \mathrm{rank}(\rho_{X,Y}) = 7$, and the ``standard'' way of showing that $\rho_{X,Y}$ is separable is to use a twirling argument \cite{HH99} that does not give a good bound on $\ell(\rho_{X,Y})$.

As an even more extreme example, consider the PCP pair $(X,Y)$ given by
\[
    X = Y = I_n + J_n \in M_n(\mathbb{C}).
\]
We noted in Example~\ref{exam:werner_ccp} that this pair is PCP, and we provided a decomposition that shows that $\ell(X,Y) \leq n$. In fact, $\ell(X,Y) = n$ since $\ell(X,Y) \geq \mathrm{rank}(X) = n$ as well. On the other hand, the corresponding CLDUI quantum state $\rho_{X,Y} \in M_n(\mathbb{C}) \otimes M_n(\mathbb{C})$ is
\[
    \rho_{X,Y} = I_n \otimes I_n + \left(\sum_{i=1}^n \mathbf{e_i} \otimes \mathbf{e_i}\right)\left(\sum_{i=1}^n \mathbf{e_i} \otimes \mathbf{e_i}\right)^*.
\]
It is straightforward to check that $\ell(\rho_{X,Y}) \geq \mathrm{rank}(\rho_{X,Y}) = n^2$, but determining whether or not this inequality is actually equality is extremely difficult. In particular, it was recently noted \cite{PPPR18} that the question of whether or not $\ell(\rho_{X,Y}) = n^2$ for this particular state $\rho_{X,Y}$ is equivalent to the well-known SIC-POVM conjecture from quantum information theory, which has resisted proof for decades (see \cite{Bus91,DPS04,RBSC04}).

The reason why it is so much easier to find a PCP decomposition of a pair $(X,Y)$ than a separable decomposition of the corresponding CLDUI state $\rho_{X,Y}$ is that the PCP decomposition of $(X,Y)$ only corresponds to a separable decomposition of \emph{some} state $\rho$ with the same entries in the non-zero positions of $\rho_{X,Y}$, but not necessarily $\rho_{X,Y}$ itself. For example, if we return to the pair $(X,Y)$ from Equation~\ref{eq:pcp_all_ones}, the PCP decomposition $\mathbf{v_1} = \mathbf{w_1} = (1,1,1)^T$ corresponds (as shown in the proof of Theorem~\ref{thm:ccp_separable}(c)) to the separable decomposition $\mathbf{v_1}\mathbf{v}_{\mathbf{1}}^* \otimes \overline{\mathbf{w_1}\mathbf{w}_{\mathbf{1}}}^*$ of the state
\[
    \rho = \left[\begin{array}{ccc|ccc|ccc}
        1 & 1 & 1 & 1 & 1 & 1 & 1 & 1 & 1 \\
        1 & 1 & 1 & 1 & 1 & 1 & 1 & 1 & 1 \\
        1 & 1 & 1 & 1 & 1 & 1 & 1 & 1 & 1 \\\hline
        1 & 1 & 1 & 1 & 1 & 1 & 1 & 1 & 1 \\
        1 & 1 & 1 & 1 & 1 & 1 & 1 & 1 & 1 \\
        1 & 1 & 1 & 1 & 1 & 1 & 1 & 1 & 1 \\\hline
        1 & 1 & 1 & 1 & 1 & 1 & 1 & 1 & 1 \\
        1 & 1 & 1 & 1 & 1 & 1 & 1 & 1 & 1 \\
        1 & 1 & 1 & 1 & 1 & 1 & 1 & 1 & 1
    \end{array}\right].
\]
The reason why this shows separability of $\rho_{X,Y}$ is then simply that $\int_U (U \otimes \overline{U}) \rho (U \otimes \overline{U})^* \, dU = \rho_{X,Y}$, and the twirling (integration) operation on the left preserves separability.

In fact, this is exactly why using PCP matrices for the separability problem is so useful in the first place. Finding a PCP decomposition of a pair $(X,Y)$ is equivalent to searching for a separable decomposition of \emph{any} (not necessarily CLDUI) state that twirls to $\rho_{X,Y}$, so it is typically possible to find a much simpler decomposition than $\rho_{X,Y}$ itself has.

\section{The Absolute Separability Problem}\label{sec:abs_sep}

A quantum state $\rho \in M_n(\mathbb{C}) \otimes M_n(\mathbb{C})$ is called \emph{absolutely separable} \cite{KZ00} if $U\rho U^*$ is separable for all unitary $U \in M_n(\mathbb{C}) \otimes M_n(\mathbb{C})$. A long-standing open question asks for a characterization of the absolutely separable states \cite{OpenProb15}, and we now use our results to make some progress on this problem. In particular, we investigate the question of whether or not the set of absolutely separable states coincides with the set of state that are \emph{absolutely PPT}: states $\rho$ with the property that $U\rho U^*$ has positive partial transpose (PPT) for all unitary matrices $U \in M_m(\mathbb{C}) \otimes M_n(\mathbb{C})$ \cite{Hil07}.

In the $n = 2$ case, the sets of absolutely separable and absolutely PPT states have been completely characterized \cite{Hil07,Joh13,VAD01}, and they both equal the set of states with eigenvalues $\lambda_1 \geq \lambda_2 \geq \lambda_3 \geq \lambda_4 \geq 0$ such that
\begin{align}\label{eq:22_abs_sep}
    \begin{bmatrix}2\lambda_4 & \lambda_3 - \lambda_1 \\ \lambda_3 - \lambda_1 & 2\lambda_2\end{bmatrix} \succeq 0.
\end{align}
Alternatively, another way of looking at the absolute separability problem when $n = 2$ is that $\rho$ is absolutely separable if and only if $U_1 \Lambda U_1^*$ is separable, where $\Lambda$ is the diagonal matrix with diagonal entries $\lambda_1, \ldots, \lambda_4$, and $U_1$ is the following unitary matrix (where we use $\cdot$ to denote entries equal to $0$):
\[
    U_1 = \begin{bmatrix}
        \cdot & \cdot & \cdot & 1 \\ 1/\sqrt{2} & \cdot & 1/\sqrt{2} & \cdot \\ -1/\sqrt{2} & \cdot & 1/\sqrt{2} & \cdot \\ \cdot & 1 & \cdot & \cdot
    \end{bmatrix}.
\]
To see the equivalence of these two characterizations of absolute separability, we can directly compute
\[
    (id \otimes T)(U_1 \Lambda U_1^*) = \frac{1}{2}\left[\begin{array}{cc|cc}
        2\lambda_4 & \cdot & \cdot & \lambda_3 - \lambda_1 \\
        \cdot & \lambda_3 + \lambda_1 & \cdot & \cdot \\\hline
        \cdot & \cdot & \lambda_3 + \lambda_1 & \cdot \\
        \lambda_3 - \lambda_1 & \cdot & \cdot & 2\lambda_2
    \end{array}\right],
\]
which has a block-diagonal form that ensures that it is positive semidefinite (and thus $U_1 \Lambda U_1^*$ is separable) if and only if the positive semidefinite condition~\eqref{eq:22_abs_sep} holds. In other words, separability of $U\rho U^*$ for \emph{all} unitary matrices $U \in M_2(\mathbb{C}) \otimes M_2(\mathbb{C})$ really only depends on separability of $U\rho U^*$ for one very special unitary.

The set of absolutely PPT states has also been completely characterized when $n \geq 3$, but the details are quite a bit more complicated. We now provide some of these details, but for a full and rigorous description, we refer the interested reader to \cite{Hil07}. We start by constructing several linear maps $L_j : M_n(\mathbb{R}) \rightarrow \mathbb{R}^{n^2}$. To illustrate how these linear maps $\{L_j\}$ are constructed, suppose $\alpha_1 \geq \alpha_2 \geq \cdots \geq \alpha_n \geq 0$ are real numbers and consider the possible orderings of their products $\alpha_j^2$ ($1 \leq j \leq m$) and $\pm\alpha_i\alpha_j$ ($1 \leq i \neq j \leq n$). For example, if $n = 2$ then the only possible ordering is $\alpha_1^2 \geq \alpha_1\alpha_2 \geq \alpha_2^2 \geq -\alpha_1\alpha_2$, whereas if $n = 3$ then there are two possible orderings:
\begin{align*}
	\alpha_1^2 & \geq \alpha_1\alpha_2 \geq \alpha_1\alpha_3 \geq \alpha_2^2 \geq \alpha_2\alpha_3 \geq \alpha_3^2 \geq -\alpha_2\alpha_3 \geq -\alpha_1\alpha_3 \geq -\alpha_1\alpha_2 \quad \text{or} \\
	\alpha_1^2 & \geq \alpha_1\alpha_2 \geq \alpha_2^2 \geq \alpha_1\alpha_3 \geq \alpha_2\alpha_3 \geq \alpha_3^2 \geq -\alpha_2\alpha_3 \geq -\alpha_1\alpha_3 \geq -\alpha_1\alpha_2.
\end{align*}

For each of these orderings, we associate a linear map $L_j : M_n(\mathbb{R}) \rightarrow \mathbb{R}^{n^2}$ by placing $\pm y_{i,j}$ into the position of $L_j(Y)$ where $\pm \alpha_i\alpha_j$ appears in the associated ordering (and actually $L_j$ is just a linear map on \emph{symmetric} matrices, not all matrices, so that we do not have to worry about distinguishing between $y_{i,j}$ and $y_{j,i}$). For example, in the $n = 2$ case, there is just one linear map $L_1 : M_2(\mathbb{R}) \rightarrow \mathbb{R}^{4}$, and it is
\[
	L_1\left(\begin{bmatrix}y_{1,1} & y_{1,2} \\ y_{1,2} & y_{2,2}\end{bmatrix}\right) = (y_{1,1},y_{1,2},y_{2,2},-y_{1,2}).
\]
Similarly, in the $n = 3$ case there are two linear maps $L_1,L_2 : M_3(\mathbb{R}) \rightarrow \mathbb{R}^{9}$, and they are
\begin{align}\begin{split}\label{eq:3x3_L_maps}
	L_1\left(\begin{bmatrix}y_{1,1} & y_{1,2} & y_{1,3} \\ y_{1,2} & y_{2,2} & y_{2,3} \\ y_{1,3} & y_{2,3} & y_{3,3}\end{bmatrix}\right) & = (y_{1,1},y_{1,2},y_{1,3},y_{2,2},y_{2,3},y_{3,3},-y_{2,3},-y_{1,3},-y_{1,2}) \quad \text{and} \\
	L_2\left(\begin{bmatrix}y_{1,1} & y_{1,2} & y_{1,3} \\ y_{1,2} & y_{2,2} & y_{2,3} \\ y_{1,3} & y_{2,3} & y_{3,3}\end{bmatrix}\right) & = (y_{1,1},y_{1,2},y_{2,2},y_{1,3},y_{2,3},y_{3,3},-y_{2,3},-y_{1,3},-y_{1,2}).
\end{split}\end{align}

The number of distinct possible orderings (and thus the number of linear maps to be considered) grows exponentially in $n$. For $n = 2, 3, 4, \ldots$, this quantity equals $1, 2, 10, 114, 2608, 107498, \ldots$, though no formula is known for computing it in general \cite{oeisA237749}.

The main result of \cite{Hil07} says that $\rho$ is absolutely PPT if and only if its eigenvalues $\lambda_1 \geq \lambda_2 \geq \cdots \geq \lambda_{n^2}$ are such that $L_j^*(\lambda_{n^2},\lambda_{n^2-1},\ldots,\lambda_{2},\lambda_{1}) \succeq O$ for all $j$. For example, when $n = 2$ this is exactly equivalent to the positive semidefinite requirement~\eqref{eq:22_abs_sep}, and if $n = 3$ then it says that $\rho$ is absolutely PPT if and only if its eigenvalues satisfy
\begin{align}\label{eq:abs_ppt_3x3}
\begin{bmatrix}
2\lambda_9 & \lambda_8 - \lambda_1 & \lambda_6 - \lambda_2 \\
\lambda_8 - \lambda_1 & 2\lambda_7 & \lambda_5 - \lambda_3 \\
\lambda_6 - \lambda_2 & \lambda_5 - \lambda_3 & 2\lambda_4
\end{bmatrix} \succeq O \quad \text{and} \quad \begin{bmatrix}
2\lambda_9 & \lambda_8 - \lambda_1 & \lambda_7 - \lambda_2 \\
\lambda_8 - \lambda_1 & 2\lambda_6 & \lambda_5 - \lambda_3 \\
\lambda_7 - \lambda_2 & \lambda_5 - \lambda_3 & 2\lambda_4
\end{bmatrix}\succeq O.
\end{align}

We can also view this characterization of absolutely PPT states as a statement that, instead of checking whether or not $(U \rho U^*)^\Gamma \succeq O$ for all unitary $U$, it suffices to just check a certain finite number of unitaries (with each unitary corresponding to one of the possible orderings discussed earlier).

For example, in the $n = 3$ case, $\rho$ is absolutely PPT if and only if $U_1 \Lambda U_1^*$ and $U_2 \Lambda U_2^*$ are PPT, where $\Lambda$ is the diagonal matrix with diagonal entries $\lambda_1, \lambda_2, \ldots, \lambda_9$, and $U_1$ and $U_2$ are the unitary matrices
\begin{align*}
    U_1 & = \begin{bmatrix}
        \cdot & \cdot & \cdot & \cdot & \cdot & \cdot & \cdot & \cdot & 1 \\
        1/\sqrt{2} & \cdot & \cdot & \cdot & \cdot & \cdot & \cdot & 1/\sqrt{2} & \cdot \\
        \cdot & 1/\sqrt{2} & \cdot & \cdot & \cdot & 1/\sqrt{2} & \cdot & \cdot & \cdot \\
        -1/\sqrt{2} & \cdot & \cdot & \cdot & \cdot & \cdot & \cdot & 1/\sqrt{2} & \cdot \\
        \cdot & \cdot & \cdot & \cdot & \cdot & \cdot & 1 & \cdot & \cdot \\
        \cdot & \cdot & 1/\sqrt{2} & \cdot & 1/\sqrt{2} & \cdot & \cdot & \cdot & \cdot \\
        \cdot & -1/\sqrt{2} & \cdot & \cdot & \cdot & 1/\sqrt{2} & \cdot & \cdot & \cdot \\
        \cdot & \cdot & -1/\sqrt{2} & \cdot & 1/\sqrt{2} & \cdot & \cdot & \cdot & \cdot \\
        \cdot & \cdot & \cdot & 1 & \cdot & \cdot & \cdot & \cdot & \cdot
    \end{bmatrix} \\
    U_2 & = \begin{bmatrix}
        \cdot & \cdot & \cdot & \cdot & \cdot & \cdot & \cdot & \cdot & 1 \\
        1/\sqrt{2} & \cdot & \cdot & \cdot & \cdot & \cdot & \cdot & 1/\sqrt{2} & \cdot \\
        \cdot & 1/\sqrt{2} & \cdot & \cdot & \cdot & \cdot & 1/\sqrt{2} & \cdot & \cdot \\
        -1/\sqrt{2} & \cdot & \cdot & \cdot & \cdot & \cdot & \cdot & 1/\sqrt{2} & \cdot \\
        \cdot & \cdot & \cdot & \cdot & \cdot & 1 & \cdot & \cdot & \cdot \\
        \cdot & \cdot & 1/\sqrt{2} & \cdot & 1/\sqrt{2} & \cdot & \cdot & \cdot & \cdot \\
        \cdot & -1/\sqrt{2} & \cdot & \cdot & \cdot & \cdot & 1/\sqrt{2} & \cdot & \cdot \\
        \cdot & \cdot & -1/\sqrt{2} & \cdot & 1/\sqrt{2} & \cdot & \cdot & \cdot & \cdot \\
        \cdot & \cdot & \cdot & 1 & \cdot & \cdot & \cdot & \cdot & \cdot
    \end{bmatrix}.
\end{align*}
The reason for this is simply that $(id \otimes T)(U_1 \Lambda U_1^*)$ equals
\[
    \frac{1}{2}\left[\begin{array}{ccc|ccc|ccc}
        2\lambda_9 & \cdot & \cdot & \cdot & \lambda_8-\lambda_1 & \cdot & \cdot & \cdot & \lambda_6-\lambda_2 \\
        \cdot & \lambda_8+\lambda_1 & \cdot & \cdot & \cdot & \cdot & \cdot & \cdot & \cdot \\
        \cdot & \cdot & \lambda_6+\lambda_2 & \cdot & \cdot & \cdot & \cdot & \cdot & \cdot \\\hline
        \cdot & \cdot & \cdot & \lambda_8+\lambda_1 & \cdot & \cdot & \cdot & \cdot & \cdot \\
        \lambda_8-\lambda_1 & \cdot & \cdot & \cdot & 2\lambda_7 & \cdot & \cdot & \cdot & \lambda_5-\lambda_3 \\
        \cdot & \cdot & \cdot & \cdot & \cdot & \lambda_5+\lambda_3 & \cdot & \cdot & \cdot \\\hline
        \cdot & \cdot & \cdot & \cdot & \cdot & \cdot & \lambda_6+\lambda_2 & \cdot & \cdot \\
        \cdot & \cdot & \cdot & \cdot & \cdot & \cdot & \cdot & \lambda_5+\lambda_3 & \cdot \\
        \lambda_6-\lambda_2 & \cdot & \cdot & \cdot & \lambda_5-\lambda_3 & \cdot & \cdot & \cdot & 2\lambda_5
    \end{array}\right]
\]
and $(id \otimes T)(U_2 \Lambda U_2^*)$ equals
\[
    \frac{1}{2}\left[\begin{array}{ccc|ccc|ccc}
        2\lambda_9 & \cdot & \cdot & \cdot & \lambda_8-\lambda_1 & \cdot & \cdot & \cdot & \lambda_7-\lambda_2 \\
        \cdot & \lambda_8+\lambda_1 & \cdot & \cdot & \cdot & \cdot & \cdot & \cdot & \cdot \\
        \cdot & \cdot & \lambda_7+\lambda_2 & \cdot & \cdot & \cdot & \cdot & \cdot & \cdot \\\hline
        \cdot & \cdot & \cdot & \lambda_8+\lambda_1 & \cdot & \cdot & \cdot & \cdot & \cdot \\
        \lambda_8-\lambda_1 & \cdot & \cdot & \cdot & 2\lambda_6 & \cdot & \cdot & \cdot & \lambda_5-\lambda_3 \\
        \cdot & \cdot & \cdot & \cdot & \cdot & \lambda_5+\lambda_3 & \cdot & \cdot & \cdot \\\hline
        \cdot & \cdot & \cdot & \cdot & \cdot & \cdot & \lambda_7+\lambda_2 & \cdot & \cdot \\
        \cdot & \cdot & \cdot & \cdot & \cdot & \cdot & \cdot & \lambda_5+\lambda_3 & \cdot \\
        \lambda_6-\lambda_2 & \cdot & \cdot & \cdot & \lambda_5-\lambda_3 & \cdot & \cdot & \cdot & 2\lambda_5
    \end{array}\right],
\]
which are positive semidefinite if and only if the matrices~\eqref{eq:abs_ppt_3x3} are positive semidefinite.

Despite this characterization of absolutely PPT states, not much is known about the set of absolutely separable states when $n \geq 3$. For example, it is not even known whether or not it equals the set of absolutely PPT states, though some evidence that suggests that these two sets may coincide was provided in \cite{AJR15}. The results of this paper provide further evidence in favor of these sets being identical, since they imply that whenever the states $(id \otimes T)(U_j \Lambda U_j^*)$ are PPT they are also necessarily separable.

To see why this is the case, we just note that if $U_j$ is one of the unitaries coming from from the characterization of absolutely PPT states in the manner described above and $\Lambda$ is diagonal with real entries in non-increasing order then $(id \otimes T)(U_j\Lambda U_j^*)$ is a CLDUI state (simply by virtue of its zero pattern) with non-positive off-diagonal entries (since the off-diagonal entries all consist of a small eigenvalue minus a larger eigenvalue). Thus $(id \otimes T)(U_j\Lambda U_j^*)$ equals its own comparison matrix $M((id \otimes T)(U_j\Lambda U_j^*))$, so Corollary~\ref{cor:ldui_state_m_matrix} tells us that $U_j\Lambda U_j^*$ is separable whenever it has PPT.

This does not quite prove that the set of absolutely separable states equals the set of absolutely PPT states, however, as we do not know whether or not $U_j\Lambda U_j^*$ being separable for all of these special unitary matrices $\{U_j\}$ means that $U\Lambda U^*$ is separable for \emph{all} unitary matrices. Nonetheless, we find the fact that these matrices are separable in all dimensions quite surprising, and this rules out the ``obvious'' method of finding a gap between absolute separability and absolute PPT, which would be to find entanglement in one of these special states of the form $U_j\Lambda U_j^*$.

\section{Conclusions and Future Work}\label{sec:conclusions}

In this work, we introduced a family of matrix pairs that we called \emph{pairwise completely positive} as a generalization of completely positive matrices. We presented numerous simple necessary and sufficient conditions that can be used to determine whether or not certain matrix pairs are PCP, and we showed that PCP pairs correspond in a one-to-one fashion with separable conjugate local diagonal unitary invariant quantum states.

The most obvious open question remaining open from this work is the absolute separability problem. It was already known that there exists a finite set of unitaries $\{U_j\}$ for which $U_j \Lambda U_j^*$ being PPT for all $j$ implies that $\Lambda$ is absolutely PPT (i.e., $U \Lambda U^*$ is PPT for \emph{all} unitaries). We showed that if $U_j \Lambda U_j^*$ is PPT then $U_j \Lambda U_j^*$ must also be separable, so the natural follow-up question is whether or not $U_j \Lambda U_j^*$ being separable for all $j$ implies that $\Lambda$ is absolutely separable (i.e., $U \Lambda U^*$ is separable for \emph{all} unitaries).

As one possible generalization of this work, one could consider the quantum states that are invariant under all \emph{real} diagonal unitary matrices (i.e., diagonal matrices with diagonal entries equal to $\pm 1$). Such states have the slightly more general form
\begin{align*}
    \rho = \sum_{i,j=1}^n x_{i,j} \mathbf{e_i}\mathbf{e}_{\mathbf{j}}^* \otimes \mathbf{e_i}\mathbf{e}_{\mathbf{j}}^* + \sum_{i \neq j=1}^n y_{i,j} \mathbf{e_i}\mathbf{e}_{\mathbf{i}}^* \otimes \mathbf{e_j}\mathbf{e}_{\mathbf{j}}^* + \sum_{i \neq j=1}^n z_{i,j} \mathbf{e_i}\mathbf{e}_{\mathbf{j}}^* \otimes \mathbf{e_j}\mathbf{e}_{\mathbf{i}}^*,
\end{align*}
which are CLDUI states exactly when $z_{i,j} = 0$ for all $i,j$. Checking separability of these states in a manner analogous to that of Theorem~\ref{thm:ccp_separable} leads to a triple of matrices $(X,Y,Z)$ with the property that $\rho$ is separable if and only if there exist matrices $V,W \in M_{n,m}(\mathbb{C})$ such that
\begin{align*}
    X = (V \odot W)(V \odot W)^*, \qquad Y = (V \odot \overline{V})(W \odot \overline{W})^*, \quad \text{and} \quad Z = (V \odot \overline{W})(V \odot \overline{W})^*.
\end{align*}
In this case, it is clear that $Z$ is positive semidefinite (in addition to $X$ being positive semidefinite and $Y$ being entrywise non-negative) and all three matrices have the same diagonal entries. However, we did not investigate further properties of these matrix triples, as our interest in pairwise completely positive matrices came primarily from the absolute separability problem, in which case the relevant states are CLDUI.\\

\noindent \textbf{Acknowledgements.} N.J.\ was supported by NSERC Discovery Grant number RGPIN-2016-04003 and O.M.\ was supported by an NSERC Undergraduate Student Research Award. The authors thank users ``GH from MO'' and ``Mark Wildon'' at MathOverflow for providing the proof of Lemma~\ref{lem:cs_stronger} \cite{MO_Lemma}.

\bibliographystyle{ieeetr}
\bibliography{bib}

\begin{thebibliography}{10}

\bibitem{BDS15}
A.~Berman, M.~Dur, and N.~Shaked-Monderer, ``Open problems in the theory of
  completely positive and copositive matrices,'' {\em Electronic Journal of
  Linear Algebra}, vol.~29, pp.~46--58, 2015.

\bibitem{Ber88}
A.~Berman, ``Complete positivity,'' {\em Linear Algebra Appl.}, vol.~107,
  pp.~57--63, 1988.

\bibitem{BS03}
A.~Berman and N.~Shaked-Monderer, {\em Completely Positive Matrices}.
\newblock World Scientific, 2003.

\bibitem{Bom12}
I.~M. Bomze, ``Copositive optimization - recent developments and
  applications,'' {\em European Journal of Operational Research}, vol.~216,
  pp.~509--520, 2012.

\bibitem{Dur10}
M.~D\"{u}r, ``Copositive programming - a survey,'' in {\em Recent Advances in
  Optimization and its Applications in Engineering} (M.~Diehl, F.~Glineur,
  E.~Jarlebring, and W.~Michiels, eds.), pp.~3--20, Springer, 2010.

\bibitem{Yu16}
N.~Yu, ``Separability of a mixture of dicke states,'' {\em Physical Review A},
  vol.~94, p.~060101(R), 2016.

\bibitem{TAQLS17}
J.~Tura, A.~Aloy, R.~Quesada, M.~Lewenstein, and A.~Sanpera, ``Separability of
  mixed {D}icke states: an {NP}-hard optimization problem,'' {\em Quantum},
  vol.~2, p.~45, 2018.

\bibitem{BCP14}
T.~Baumgratz, M.~Cramer, and M.~B. Plenio, ``Quantifying coherence,'' {\em
  Physical Review Letters}, vol.~113, p.~140401, 2014.

\bibitem{DG14}
P.~J.~C. Dickinson and L.~Gijben, ``On the computational complexity of
  membership problems for the completely positive cone and its dual,'' {\em
  Computational Optimization and Applications}, vol.~57, pp.~403--415, 2014.

\bibitem{MM62}
J.~E. Maxfield and H.~Minc, ``On the matrix equation $x^\prime x = a$,'' {\em
  Proc. Edinburgh Math. Soc. Series II}, vol.~13, pp.~125--129, 1962.

\bibitem{Kay87}
M.~Kaykobad, ``On nonnegative factorization matrices,'' {\em Linear Algebra
  Appl.}, vol.~96, pp.~27--33, 1987.

\bibitem{DJL94}
J.~H. Drew, C.~R. Johnson, and R.~Loewy, ``Completely positive matrices
  associated with {M}-matrices,'' {\em Linear and Multilinear Algebra},
  vol.~37, pp.~303--310, 1994.

\bibitem{Wer89}
R.~F. Werner, ``Quantum states with {E}instein--{P}odolsky--{R}osen
  correlations admitting a hidden-variable model,'' {\em Physical Review A},
  vol.~40, pp.~4277--4281, 1989.

\bibitem{Gha10}
S.~Gharibian, ``Strong {NP}-hardness of the quantum separability problem,''
  {\em Quantum Inf. Comput.}, vol.~10, pp.~343--360, 2010.

\bibitem{Gur03}
L.~Gurvits, ``Classical deterministic complexity of {E}dmonds' problem and
  quantum entanglement,'' in {\em Proceedings of the Thirty-Fifth Annual ACM
  Symposium on Theory of Computing}, pp.~10--19, 2003.

\bibitem{Per96}
A.~Peres, ``Separability criterion for density matrices,'' {\em Physical Review
  Letters}, vol.~77, pp.~1413--1415, 1996.

\bibitem{HHH96}
M.~Horodecki, P.~Horodecki, and R.~Horodecki, ``Separability of mixed states:
  Necessary and sufficient conditions,'' {\em Phys. Lett. A}, vol.~223,
  pp.~1--8, 1996.

\bibitem{CW03}
K.~Chen and L.-A. Wu, ``A matrix realignment method for recognizing
  entanglement,'' {\em Quantum Inf. Comput.}, vol.~3, pp.~193--202, 2003.

\bibitem{Rud00}
O.~Rudolph, ``A separability criterion for density operators,'' {\em J. Phys.
  A: Math. Gen.}, vol.~33, pp.~3951--3955, 2000.

\bibitem{GT09}
O.~G\"{u}hne and G.~Toth, ``Entanglement detection,'' {\em Physics Reports},
  vol.~474, pp.~1--75, 2009.

\bibitem{HHH09}
R.~Horodecki, P.~Horodecki, M.~Horodecki, and K.~Horodecki, ``Quantum
  entanglement,'' {\em Reviews of Modern Physics}, vol.~81, pp.~865--942, 2009.

\bibitem{HH99}
M.~Horodecki and P.~Horodecki, ``Reduction criterion of separability and limits
  for a class of distillation protocols,'' {\em Phys. Rev. A}, vol.~59,
  pp.~4206--4216, 1999.

\bibitem{PPPR18}
S.~K. Pandey, V.~I. Paulsen, J.~Prakash, and M.~Rahaman, ``Entanglement
  breaking rank,'' {\em E-print: arXiv:1805.04583 [quant-ph]}, 2018.

\bibitem{Bus91}
P.~Busch, ``Informationally complete sets of physical quantities,'' {\em
  Internat. J. Theoret. Phys.}, vol.~30, no.~9, pp.~1217--1227, 1991.

\bibitem{DPS04}
G.~M. D'Ariano, P.~Perinotti, and M.~F. Sacchi, ``Informationally complete
  measurements and group representation,'' {\em Journal of Optics B: Quantum
  and Semiclassical Optics}, vol.~6, no.~6, pp.~S487--S491, 2004.

\bibitem{RBSC04}
J.~M. Renes, R.~Blume-Kohout, A.~J. Scott, and C.~M. Caves, ``Symmetric
  informationally complete quantum measurements,'' {\em J. Math. Phys.},
  vol.~45, no.~6, pp.~2171--2180, 2004.

\bibitem{KZ00}
M.~Ku{\'s} and K.~{\.Z}yczkowski, ``Geometry of entangled states,'' {\em Phys.
  Rev. A}, vol.~63, p.~032307, 2001.

\bibitem{OpenProb15}
E.~Knill, ``Separability from spectrum (problem 15).'' Published electronically
  at \url{https://oqp.iqoqi.univie.ac.at/separability-from-spectrum/}, 2003.

\bibitem{Hil07}
R.~Hildebrand, ``Positive partial transpose from spectra,'' {\em Phys. Rev. A},
  vol.~76, p.~052325, 2007.

\bibitem{Joh13}
N.~Johnston, ``Separability from spectrum for qubit--qudit states,'' {\em Phys.
  Rev. A}, vol.~88, p.~062330, 2013.

\bibitem{VAD01}
F.~Verstraete, K.~Audenaert, and B.~D. Moor, ``Maximally entangled mixed states
  of two qubits,'' {\em Phys. Rev. A}, vol.~64, p.~012316, 2001.

\bibitem{oeisA237749}
N.~Johnston, ``The {O}n-{L}ine {E}ncyclopedia of {I}nteger {S}equences.''
  A237749, February 2014.
\newblock The number of possible orderings of the real numbers $x_i x_j$ ($i
  \leq j$), subject to the constraint that $x_1 > x_2 > ... > x_n > 0$.

\bibitem{AJR15}
S.~Arunachalam, N.~Johnston, and V.~Russo, ``Is absolute separability
  determined by the partial transpose?,'' {\em Quantum Information and
  Computation}, vol.~15, pp.~0694--0720, 2015.

\bibitem{MO_Lemma}
N.~Johnston, ``A strengthening of the {C}auchy--{S}chwarz inequality.''
  MathOverflow.
\newblock \url{https://mathoverflow.net/q/301844}.

\end{thebibliography}

\section*{Appendix: A Stronger Cauchy--Schwarz Inequality}

We now prove Inequality~\eqref{eq:stronger_cs_ineq}, which is required to complete the proof of Theorem~\ref{thm:ccp_properties}(e). In particular, we have the following:

\begin{lemma}\label{lem:cs_stronger}
    Suppose $\mathbf{v},\mathbf{w} \in \mathbb{C}^n$. Then
    \begin{align*}
        0 \leq \|\mathbf{v\odot \overline{v}}\|\|\mathbf{w\odot \overline{w}}\| - \langle \mathbf{v\odot \overline{v}}, \mathbf{w\odot \overline{w}} \rangle \leq \|\mathbf{v}\|^2\|\mathbf{w}\|^2 - |\langle \mathbf{v}, \mathbf{w} \rangle|^2.
    \end{align*}
\end{lemma}
\begin{proof}
    The left inequality follows directly from Cauchy--Schwarz, so we focus on the right inequality. We note that without loss of generality we can assume that $\mathbf{v}$ and $\mathbf{w}$ have real non-negative entries, since multiplying any of the entries of these vectors by a complex phase $e^{i\theta}$ does not change the terms $\|\mathbf{v\odot \overline{v}}\|\|\mathbf{w\odot \overline{w}}\|$, $\langle \mathbf{v\odot \overline{v}}, \mathbf{w\odot \overline{w}}\rangle$, or $\|\mathbf{v}\|^2\|\mathbf{w}\|^2$, and the triangle inequality tells us that the term $|\langle \mathbf{v}, \mathbf{w} \rangle|^2$ is largest (and thus the inequality provided by this lemma is strongest) when $\mathbf{v}$ and $\mathbf{w}$ are real and entrywise non-negative.
    
    We now make use of Lagrange's identity to rewrite our statement. Lagrange's identity states:
    
\begin{align*}
\sum_{i< j}^{n} (a_{i}b_j - a_j b_i)^2= \sum_{i = 1}^n a_{i}^2\sum_{i = 1}^n b_{i}^2 - \Big(\sum_{i = 1}^na_ib_i\Big)^2                     
\end{align*}
The inequality we want to prove reads:
\begin{align*}
\left(\sum_i v_i^4\right)^{1/2}\left(\sum_i w_i^4\right)^{1/2}-\sum_i v_i^2 w_i^2\leq 
\left(\sum_i v_i^2\right)\left(\sum_i w_i^2\right)-\left(\sum_i v_i w_i\right)^2.
\end{align*}

\noindent
Rewriting the right hand side using Lagrange's identity we get:
\begin{align*}
\left(\sum_i v_i^4\right)^{1/2}\left(\sum_i w_i^4\right)^{1/2}-\sum_i v_i^2 w_i^2 \leq \sum_{i<j} (v_i{w_j}-v_j{w_i})^2.
\end{align*}
Equivalently,
\begin{align*}
\left(\sum_i v_i^4\right)\left(\sum_i w_i^4\right)\leq \left(\sum_i v_i^2 w_i^2+\sum_{i<j} (v_i{w_j}-v_j{w_i})^2\right)^2.
\end{align*}
\noindent
Rewriting the left hand side also by means of Lagrange's identity, the desired inequality reads: 
\begin{align*}
\Bigg(\sum_i  v_i^2 w_i^2\Bigg)^2+\sum_{i<j}(v_i^2w_j^2-v_j^2w_i^2)^2\leq\left(\sum_i v_i^2 w_i^2+\sum_{i<j}(v_iw_j-v_j w_i)^2\right)^2.
\end{align*}
Expanding the square and simplifying, we get the equivalent statement:
\begin{align}
\sum_{i<j}(v_i^2w_j^2-v_j^2w_i^2)^2\leq 2\left(\sum_i v_i^2w_i^2\right)\sum_{i<j} (v_i{w_j}-v_j{w_i})^2+\left(\sum_{i<j} (v_i{w_j}-v_j{w_i})^2\right)^2. \label{eq:thing}
\end{align}
Now, to actually prove the inequality holds, we make the observation that: 
\begin{align*}
&2\sum_{i<j}(v_iw_i-v_jw_j)^2(v_iw_j-v_jw_i)^2=\\
&2\sum_{i<j}(v_i^2w_i^2+v_j^2w_j^2)(v_iw_j-v_j w_i)^2+\sum_{i<j}(v_iw_j-v_j w_i)^4 - \sum_{i<j}(v_i^2w_j^2-v_j^2w_i^2)^2.
\end{align*}
Since this quantity is a sum of squares, it is positive. We thus have: 
\begin{align*}
\sum_{i<j}(v_i^2w_j^2-v_j^2w_i^2)^2 \leq  2\sum_{i<j}(v_i^2w_i^2+v_j^2w_j^2)(v_iw_j-v_j w_i)^2+\sum_{i<j}(v_iw_j-v_j w_i)^4.
\end{align*}
Then, since the right hand side of the last inequality is less than the right hand side of inequality \eqref{eq:thing}, we know it is stronger and so we have the following:
\begin{align*}
\sum_{i<j}(v_i^2w_j^2-v_j^2w_i^2)^2&\leq 2\sum_{i<j}(v_i^2w_i^2+v_j^2w_j^2)(v_iw_j-v_j w_i)^2+\sum_{i<j}(v_iw_j-v_j w_i)^4 \\ 
&\leq  2\left(\sum_i v_i^2w_i^2\right)\sum_{i<j} (v_i{w_j}-v_j{w_i})^2+\left(\sum_{i<j} (v_i{w_j}-v_j{w_i})^2\right)^2,
\end{align*}
and we have shown the desired inequality holds.
\end{proof}
\end{document}